\documentclass[16pt]{article}
\usepackage{bm}
\usepackage[toc,page]{appendix}
\usepackage{makeidx}
\usepackage{amssymb}
\usepackage{stmaryrd}
\usepackage{placeins}  
\usepackage{graphics}
\usepackage{subfig}
\usepackage[utf8]{inputenc}
\usepackage{graphicx}

\usepackage{latexsym}

\usepackage{amssymb}
\usepackage{graphics}
\usepackage{amsmath}

\usepackage{xcolor}
\usepackage{blindtext} 

\newtheorem{theorem}{\bf Theorem}
\newenvironment{proof}{\par \noindent {\bf Proof: }}{\begin{flushright}
$\Box$\end{flushright}\par \noindent}
\newtheorem{definition}[theorem]{\bf Definition}
\newtheorem{remark}[theorem]{\bf Remark}

\newtheorem{lemma}[theorem]{\bf Lemma}
\newtheorem{proposition}[theorem]{\bf Proposition}

\newtheorem{example}[theorem]{\bf Example}

\usepackage{adjustbox}
\usepackage{graphicx}

 \DeclareOldFontCommand{\bf}{\normalfont\bfseries}{\mathbf}
\usepackage{epigraph}
\usepackage{tikz}
\usepackage[tikz]{bclogo}
\usepackage{algorithm}
\usepackage{algorithmic}

\newcommand{\pol}{\sqsubseteq}

\newcommand{\qed}{\nobreak \ifvmode \relax \else
      \ifdim\lastskip<1.5em \hskip-\lastskip
      \hskip1.5em plus0em minus0.5em \fi \nobreak
      \vrule height0.75em width0.5em depth0.25em\fi}

\usetikzlibrary{calc}
\newcounter{nodeidx}
\makeatletter
\newcommand{\mypm}{\mathbin{\mathpalette\@mypm\relax}}
\newcommand{\@mypm}[2]{\ooalign{%
  \raisebox{.1\height}{$#1+$}\cr
  \smash{\raisebox{-.6\height}{$#1-$}}\cr}}
\makeatother

\newcommand{\para}{\parallel}

\makeindex
\begin{document}
\title{Entropy conservation for comparison-based algorithms\footnote{This work was completed during a Fulbright Scholarship (April 5 - Aug 31, 2019) at Stanford's Computer Science theory group, research host D. Knuth. The author is grateful for discussions with D. Knuth, V. Pratt and M. Fiore.} }
\author{M. P. Schellekens\footnote{Associate Professor, University College Cork, Department of Computer Science, Email: m.schellekens@cs.ucc.ie} }
\date{}
\maketitle
\begin{abstract} Comparison-based algorithms are algorithms for which the execution of each operation is solely based on the outcome of a series of comparisons between elements \cite{knu}. Typical examples include most sorting algorithms\footnote{Such as Bubblesort, Insertionsort, Quicksort, Mergesort, \ldots \cite{knu}}, search algorithms\footnote{Such as Quickselect \cite{knu}}, and more general algorithms such as Heapify which constructs a heap data structure from an input list \cite{knu}. \emph{Comparison-based computations can be naturally represented via the following computational model} \cite{sch1}: (a) model data structures as partially-ordered finite sets; (b) model data on these by topological
sorts\footnote{We use the computer science terminology for this notion. In mathematics the notion of a topological sort is referred to as a \emph{linear extension} of a partial order.}; (c) considering computation states as finite multisets of such data;
(d) represent computations by their induced transformations on states. 

In this view, an abstract specification of a sorting algorithm has input state given by any possible permutation of a finite set of elements (represented, according to (a) and (b), by a discrete partially-ordered set
together with its topological sorts given by all permutations) and output
state a sorted list of elements (represented, again according to (a) and (b),
by a linearly-ordered finite set with its unique topological sort). 

Entropy is a measure of ``randomness'' or ``disorder.'' Based on the computational model, we introduce an entropy conservation result for comparison-based algorithms: \emph{``quantitative order gained is proportional to positional order lost.''} Intuitively, the result bears some relation to the messy office argument advocating a chaotic office where nothing is in the right place yet each item's place is known to the owner, over the case where each item is stored in the right order and yet the owner can no longer locate the items. Formally, we generalize the result to the class of data structures representable via series-parallel partial orders--a well-known computationally tractable class \cite{moh}. The resulting ``denotational" version of entropy conservation will be extended in follow-up work to an ``operational" version for a core part of our computational model. 
 \end{abstract}

\section{Introduction}
Our work investigates properties of functions arising in computation \emph{when} complexity is taken into account. For traditional denotational semantics the main property of input-output functions is that of Scott-continuity \cite{stoy}. When studying semantics of programming languages that are required to reflect the meaning of programs (intuitively the input-output relation) \emph{and} also the complexity (the efficiency measured in running-time), more refined models are required. Such models have been studied in quantitative domain theory \cite{sch4}. The story however is far from over, since even though the theory of quantitative domains has matured at the model level, the situation at the programming language level is quite different from traditional semantics. Traditional denotational semantics relies on the \emph{compositionality} of the determination of meaning\footnote{The meaning of the sequential execution of two programs is the functional composition of their meanings.}. The complexity measure of worst-case running time is inherently non-compositional.  The complexity measure of average-case time is compositional \cite{sch1} but does not support a \emph{computation} of the compositional outcomes of complexities since input-distributions cannot, in general, be feasibly tracked throughout computations. This has led to the development of a special purpose programming language MOQA supporting a compositional determination of average-case time \cite{sch1}. A key aspect of MOQA is that its operations support ``global state preservation", which in turns guarantees a compositional determination of the average-case complexity of MOQA-algorithms. MOQA-algorithms are comparison-based. Here we continue the investigation of global state preservation and show that for the general case of comparison-based algorithms this notion can be refined to a novel notion of entropy conservation. 

\section{Basic notions}
\subsection{Orders, data structures and sorting}
We assume familiarity with the standard notion of a \emph{partial order}, including related concepts such as extremal elements (minimal and maximal elements), a Hasse diagram, a topological sort (aka linear extension of a partial order), and a linear (or total) order. 

\emph{Partial orders are implicitly assumed to be finite} unless they are clear from the context to be infinite, as is the case for the standard linear order over the set of integers. Figures displaying partial orders use a Hasse diagram representation of the transitive-reflexive reduction of the order. Partial orders are denoted by a pair $(X,\sqsubseteq)$ or using Greek letters, $\alpha,\beta, \ldots$ The size $|\alpha|$ of a partial order $\alpha = (X,\sqsubseteq)$ is the cardinality of the underlying set $X$. Specific partial orders are denoted by capital Greek or Roman letters. $\Delta_n$ denotes a discrete order of size $n$ and $L_n$ denotes a linear order of size $n$. The underlying finite set $X$ of a partial order $(X,\sqsubseteq)$ of size $n$ is typically enumerated as $X = \{x_1,\ldots,x_n\}$, using indices $i$ to enumerate elements $x_i$ of the set $X$. 

We will represent data structures via finite partial orders and adopt the following convention: \emph{elements of data structures, such as, say, lists, are enumerated starting from position 1 rather than from position 0} (as would be customary in computer science).

We assume familiarity with basic \emph{data structures} such as a tree, a complete binary tree (in which each parent node has exactly two children and the leaves all have the same path-length counting form the root), and the heap data structure (a complete binary tree-structure possibly with some leaves removed in right-to-left order and labelled with integers such that each parent label is larger than the label(s) of its child(ren)). We also assume familiarity with the notion of a \emph{comparison-based algorithm} \cite{knu}.


\emph{Multisets} are set-like structures in which order plays no role but in which, contrary to sets, duplicate elements are allowed and accounted for via multiplicities. 


Finally, we recall that the \emph{complement of a graph} $G_1$ is a graph $G_2$ on the same vertices such that two distinct vertices of $G_2$ are adjacent if and only if they are not adjacent in $G_1$.


\subsection{Topological sorts, state space and root states}
We recall that partial orders are considered to be finite in this presentation unless otherwise stated.



%


\begin{definition} \label{topologicalsortdefinition} Given a (finite) partial order $\alpha$ and a linearly ordered countable set $(\mathcal{L},\leq)$, referred to as the \textbf{label set}. A \textbf{labeling} $l$ of the order $\alpha$ is an increasing injection from $\alpha$ into $(\mathcal{L},\leq)$. Note that by definition, labelings never use duplicate labels, i.e. repeated labels. This corresponds to the standard assumption in algorithmic time analysis where, to simplify the analysis, the data, such as lists, are assumed to have distinct elements\footnote{Repeated labels can be catered for in an analysis. The details are technical \cite{sch1}.}. We adopt this convention in our computational model. 

A \textbf{topological sort} of a finite partial order $\alpha$ is a pair $(\alpha,l)$ consisting of the partial order $\alpha$ and a \textbf{labeling} $l$. $Top_{\mathcal{L}}(\alpha)$ denotes the set of all topological sorts $(\alpha,l)$ such that $range(l) \subseteq \mathcal{L}$. In other words, $Top_{\mathcal{L}}(\alpha)$ is the set of all topological sorts using labels from the given label set $\mathcal{L}$. In examples we will typically take $\mathcal{L}$ to be the positive integers, but the set could be any countable linear order, e.g. the words of the English alphabet equipped with the lexicographical order. 

Note that \textbf{a labeling $l$ of a topological sort $(\alpha,l)$} where $\alpha = (X, \sqsubseteq)$ and $X = \{x_1,\ldots,x_n\}$ \textbf{is determined by a permutation $\sigma$} on $\{1,\ldots,n\}$. The arguments of $\sigma$ are the indices of the elements $x_i$ and the values $\sigma$ takes are the ranks of the labels $l(x_i)$ (taken in the range of $l$). For instance, the topological sort over the linear order on the set $X = \{x_1,x_2,x_3\}$ determined by the labeling $l$ taking the values $l(x_1) = 5, l(x_2) = 2, l(x_3) = 7$ is the permutation $\sigma = {2 \,\, 1 \,\, 3 \choose 1 \,\,  2  \,\, 3  }$. 


\end{definition}

Figure \ref{Fig: topological-sort-example-alt} displays four topological sorts, marked I, II, III and IV, for a partial order that has a Hasse diagram forming a binary tree of size 4. The label set $\mathcal{L}$ is the set of positive integers. The four topological sorts are examples of \emph{heap data structures} \cite{knu}. 
\begin{definition}
\label{Def:LPOIso}
Two topological sorts $(\alpha,l_{1})$ and $(\alpha,l_2)$ are \textbf{isomorphic} exactly when for all $x, y\in X$, $l_{1}(x) \leq l_{1}(y)$ if and only if $l_{2}(x)\leq l_{2}(y).$ In other words \textbf{the labels of the topological sort share the same relative order}. For instance, consider the two topological sorts determined by the lablelings $(x_1: 5,x_2: 3, x_3: 1)$ and $(x_1: 4,x_2: 2, x_3: 0)$ of the discrete order $\Delta_3$ of size $3$ over the set $X = \{x_1,x_2,x_3\}$. The topological sorts are isomorphic and represent two unordered (reverse sorted) lists of size $3$. \textbf{Equivalently,} the topological sorts $(\alpha,l_{1})$ and $(\alpha,l_2)$ are isomorphic when: for all $x\in X$, the rank of $l_{1}(x)$ in the range $l_1(X)$ is equal to the rank of $l_{2}(x)$ in the range $l_2(X)$. In other words, \textbf{across topological sorts, labels of the same element must have identical rank}. Given a topological sort $(\alpha,l)$, where $\alpha = (X,\sqsubseteq)$, then its \textbf{root state} $l'$ is obtained by replacing each label of $l$ by its rank in $l(X)$. Root states are exactly the non-ismorphic topological sorts over an order $\alpha$ of size $|\alpha| = n$ that use labels from the set $\{1, \ldots, n\}$ only. \textbf{Two topological sorts} hence \textbf{are equivalent iff they share the same root state.}  The root state of the isomorphic topological sorts $(x_1: 5,x_2: 3, x_3: 1)$ and $(x_1: 4,x_2: 2, x_3: 0)$ of the discrete order $\Delta_3$ of size $3$ is the permutation-labeling $(x_1: 3,x_2: 2, x_3: 1)$. Root states are labelings that can be identified with the permutations that determine the labeling. 
\end{definition}

Figure \ref{Fig: topological-sort-example-alt} displays four topological sorts, I, II, III and IV, two of which are isomorphic (I and II). Their root state are illustrated via the topological sorts V, VI and VII, where the isomorphic topological sorts I and II share the same root state V. 

\begin{definition} The topological sorts of a finite partial order $\alpha$ are identified up to labeling isomorphism. The resulting quotient, denoted $R(\alpha)$, is called the \textbf{state space} of the partial order. With abuse of notation we denote elements of a state space by canonical \emph{representatives} of these equivalence classes (as opposed to the equivalence classes): given a partial order $\alpha$ of size $n$, then its \emph{state space} $R(\alpha)$ consists of the finitely many \emph{root states} of topological sorts over the order, which represent the finitely many states non-isomorphic topological sorts can occur in. 
\end{definition}

\begin{example}\label{state space} (Root states and state space) The label set $\mathcal{L}$ is the set of positive integers. 

\noindent {\bf a) State space representing unordered lists of size 4} \\
The discrete order implies no conditions on labels of its topological sorts. For a discrete order $\Delta_n$ of size $n$, the state space $R(\Delta_n)$ hence corresponds to the set of $n!$ permutations of size $n$. 

Consider the case of a discrete order $(X,\Delta_3)$ over a set $X = \{x_1,x_2,x_3\}$. The state space $R(\Delta_3)$ consists of the root states, i.e. the topological sorts using labels from the set $\{1,2,3\}$ only, given by the following labelings determining each such topological sort: \begin{center} $(x_1: 1, x_2: 2, x_3: 3), (x_1: 1, x_2: 3, x_3: 2), (x_2: 2, x_2: 1, x_3: 3),$ \\
$(x_1: 2, x_2: 3, x_3: 1), (x_1: 3, x_2: 1, x_3: 2), (x_1: 3, x_2: 2, x_3: 1).$ \end{center} In other words, the state space $R(X,\Delta_3)$ corresponds to the 3! permutations of size 3: \\
$\{(1,2,3),(1,3,2),(2,1,3),(2,3,1),(3,1,2),(3,2,1)\},$ representing the unordered lists of size 3. \\

\noindent {\bf b) State space representing a heap data structure of size 4} \\
Consider the four topological sorts I, II, III and IV over the order determined by the Hasse Diagram in Figure \ref{Fig: topological-sort-example-alt} and their root states V, VI and VII. These topological sorts V, VI and VII happen to form the only possible root states for this order. The set $\{V,VI,VII\}$ forms the state space of this order, representing exactly the distinct (root) states that heap data structures of size 4 can occur in. 

\begin{figure}[h]
\centering
\includegraphics[height= 7cm]{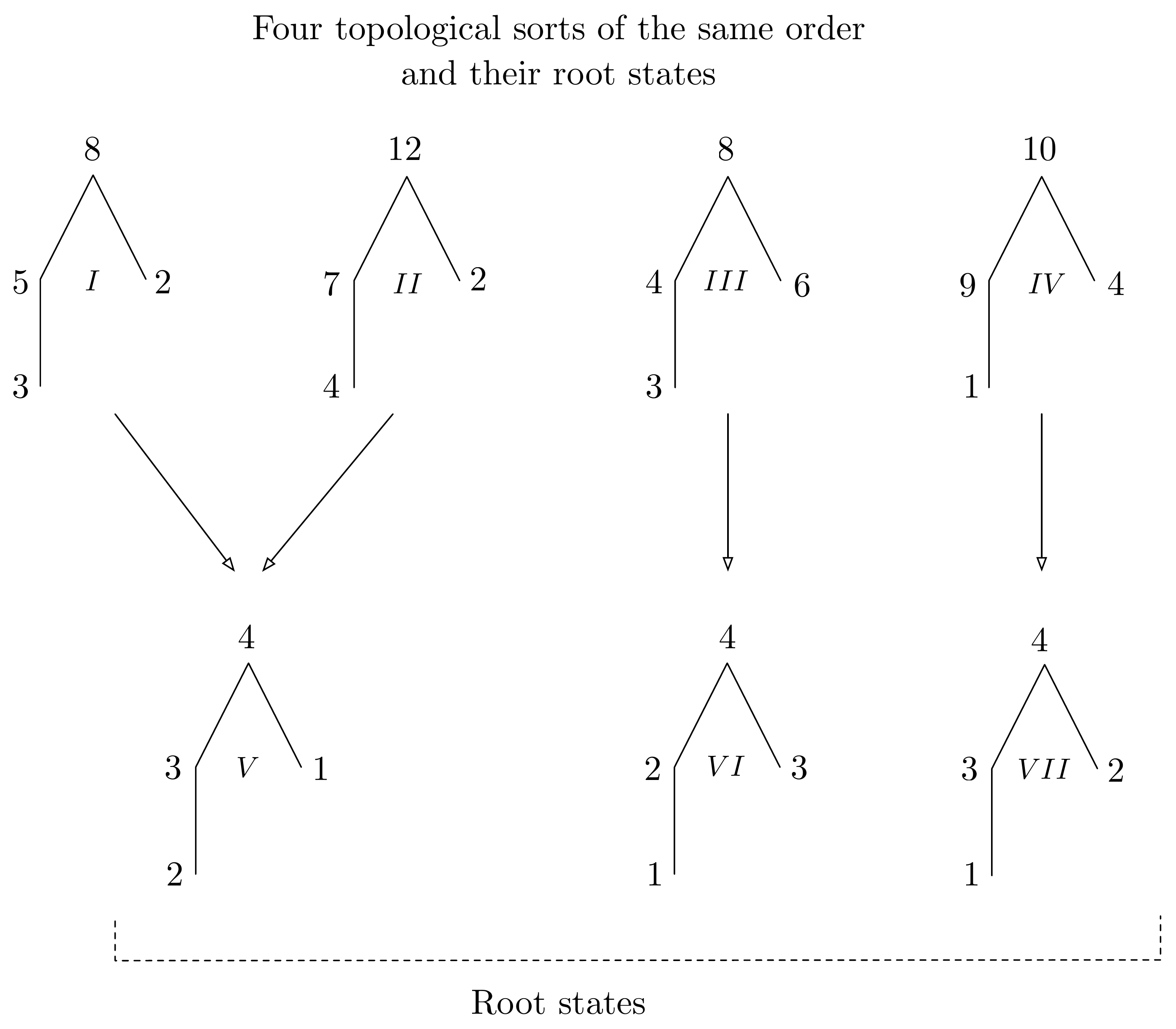}
\caption{Four topological sorts I, II, III, and IV over the same order. Topological sorts I and II are isomorphic and hence share the same root state V.  \label{Fig: topological-sort-example-alt}}
\end{figure}
\end{example}

\begin{definition} A \textbf{global state} $R$ is a finite multiset of state spaces, i.e. $R$ is of the form: $$\{(R(\alpha_1),K_1),\ldots,(R(\alpha_l),K_l)\}$$ \end{definition}
The orders $\alpha_1,\ldots, \alpha_k$ indicate that the input data structure has been transformed to several output data structures, represented by different orders\footnote{\cite{sch1} provides examples of transformations leading to different orders, e.g. Quicksort's split operation.}. Each of the state spaces $R(\alpha_i)$ reflects that output states, when they are topological sorts of the given order $\alpha_i$, have root states over $\alpha_i$. 


\subsection{The four-part model}
As indicated in the abstract, the computational model for modular time analysis of comparison-based algorithms \cite{sch1,sch3} consists of four parts: 

\begin{itemize}\item{(a) modelling data structures as
partially-ordered finite sets;} \vspace*{-1.5 mm}
\item{(b) modelling data on these by topological
sorts;}\vspace*{-1.5 mm}
\item{(c) considering computation states as finite multisets of such data (aka ``global states'');}\vspace*{-1.5 mm}
\item{(d) analysing algorithms by their induced transformations on global states.}
\end{itemize}

In this view, an abstract specification of a sorting algorithm has input state given by any possible permutation of a finite set of elements (represented, according to (a) and (b), by a discrete partially-ordered set
together with its topological sorts given by all permutations) and output
state a sorted list of elements (represented, again according to (a) and (b),
by a linearly-ordered finite set with its unique topological sort).

\begin{example} The (unordered) input lists of size 2 of a sorting algorithm are modelled by the topological sorts of a discrete order of size 2 over the set of elements $\{x_1,x_2\}$: $$\{(x_1: 1,x_2:2),(x_1:2,x_2:1)\}$$

\noindent The function values $s(x_1) = 2$ and $s(x_2) = 1$ of, say, the topological sort $s = (x_1:2,x_2:1)$ (over the discrete order of size 2) are referred to as \emph{labels} and correspond, for the case of list data structure, to the list's elements. The location $i$ of a label $a$ for which $s(x_i) = a$ is referred to as the label's \emph{index} and corresponds to the location of an element in a list. The list $(2,1)$ contains the element 2 in position 1 (i.e., the index of $x_1$) and the element 1 in position 2 (i.e. the index of $x_2$). We say that the index of label 2 is 1 and the index of label 1 is 2 for this topological sort. The 2! topological sorts $\{(x_1: 1,x_2:2),(x_1:2,x_2:1)\}$ consist of 2 permutations representing the unordered lists $(1,2)$ and $(2,1)$. These topological sorts form the \emph{``root states''} that lists of size 2 (with distinct elements) can occur in. Indeed, a list of size 2 is either sorted, represented by $(1,2)$, or reverse sorted, represented by $(2,1)$. Together, these topological sorts form a set referred to as the \emph{``state space''}\footnote{A state space intuitively serves to represent the uniform distribution over the data: each of the infinitely many possible input lists $(a,b)$ of size 2 (with distinct elements) is assumed to occur with equal probability in one of the two root states of the state space. This interpretation serves to underpin the complexity analysis of algorithms, which is the topic of \cite{sch1} and will not be considered here.}. 
\end{example}

\begin{example} \textbf{Trivial sort of lists of size 2} \label{trivial} \\
Computations will transform topological sorts to new topological sorts. All computations will be based on comparisons. For instance, a sorting algorithm, be it a very primitive one that operates only over lists of size 2, can execute a single comparison of the two elements of the list (the labels of the corresponding topological sort), followed by a swap in case the labels are out of order. Such an algorithm leaves the topological sort $(x_1: 1,x_2:2)$ unchanged and transform the topological sort $(x_1:2,x_2:1)$ via a single swap to the topological sort $(x_1:1,x_2:2)$. Sorting, in this model, produces the unique topological sort over the \emph{linear} order, as illustrated in Figure \ref{Fig:sorting-lists-size-two}.

\begin{figure}[h]
\centering
\includegraphics[height=3.4cm]{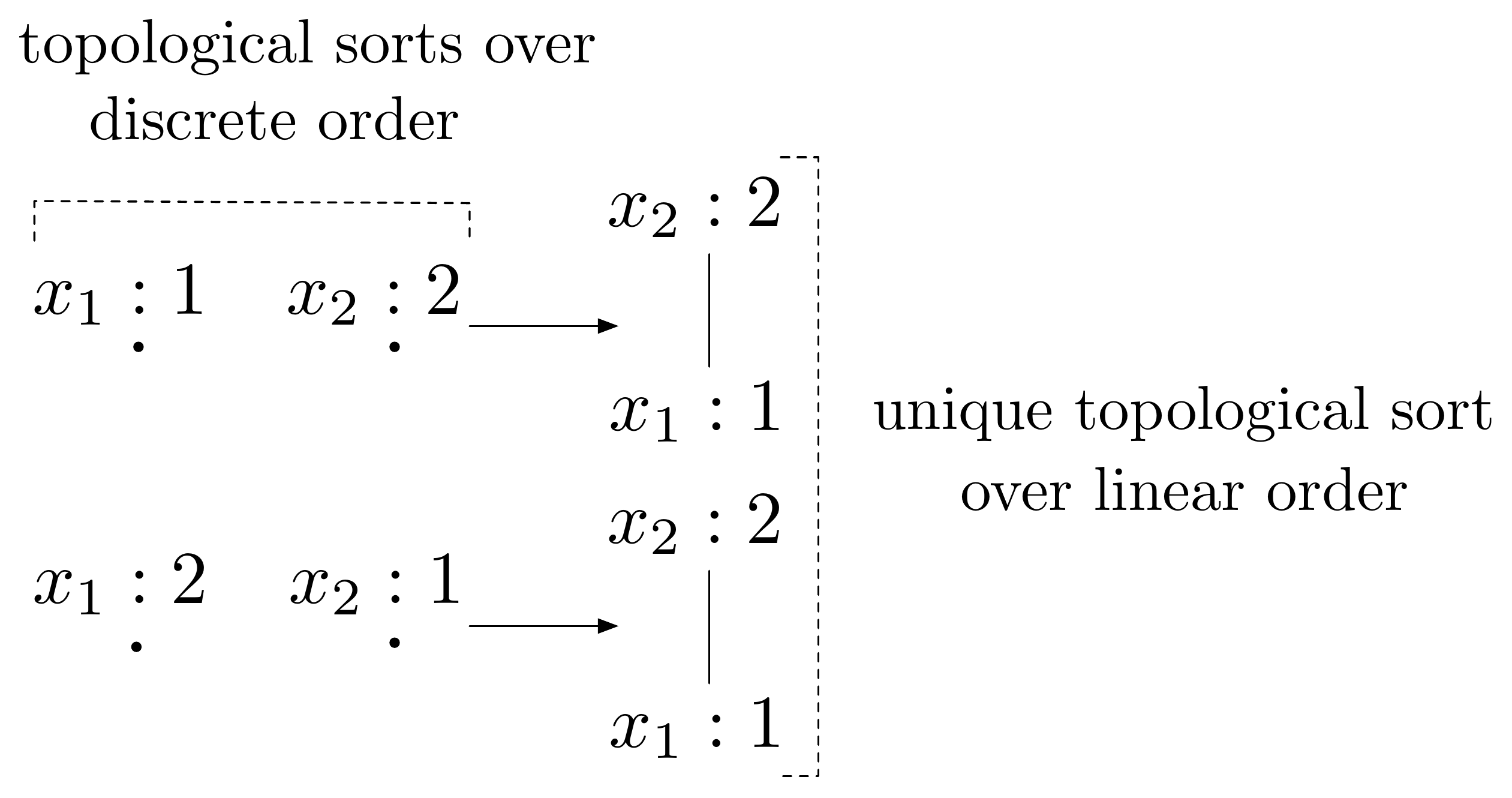}
\caption{Sorting: transforming a topological sort of the \emph{discrete order} into the unique topological sort of the \emph{linear order} for lists of size 2 \label{Fig:sorting-lists-size-two}}
\end{figure}

The transformation changes the multiset of the input state space (over the discrete order) $$\{(\{(x_1: 1,x_2:2),(x_1:2,x_2:1)\},1)\}$$ to the multiset of the output state space (over the linear order) $\{(\{(x_1: 1,x_2:2)\},2)\}.$ \\
Such multisets are referred to as ``global states" of the data under consideration.

\end{example}


We focus on the particular case of comparison-based sorting algorithms to illustrate our model\footnote{Note that the partial orders model \emph{implicit} data structures. Readers who wish to focus on the mathematical presentation, as opposed to implementation details, are advised to skip this comment on first reading. Part (a) of the model description stipulates that we use finite partial orders to represent ``data structures". This order may be \emph{implicitly} or \emph{explicitly} represented in the output data structure depending on the implementation. For instance, in the case of a heap-formation from a unordered input list, the computation can establish the heap explicitly by transforming the input list into a binary tree data structure satisfying the heap property, and constitutes a de facto heap. Alternatively, the algorithm may take an input list and retain the list data structure for its outputs. Elements of the input list will be reorganized \emph{in place}, i.e. a new list will be produced, for which the elements \emph{satisfy} a heap structure. The tree-structure underlying this heap remains an ``implicit" part of the implementation. For all purposes the algorithm makes use of the heap-structure intended by the programmer, but the data structure remains a list at all times during the computation. This is for instance the case for the ``Heapify" process in traditional (in-place) Heapsort \cite{knu}. We refer to the heap-structure in that case as the ``\emph{implicit data structure}'' and the list data structure as the ``\emph{explicit data structure}''. These may coincide or not depending on the implementation. In our context, \emph{partial orders model the \emph{implicit} data structure}. }.

 
\subsection{A basic example of the computational model: sorting algorithms}
We illustrate (a), (b), (c), and (d) of the model for a sorting algorithm operating over lists of size $n$. 

\subsubsection{Orders and topological sorts, parts (a) and (b)} 
For sorting algorithms, inputs are list data structures, represented by finite discrete orders. The elements of the order are labelled with positive integers drawn from a linearly ordered label set $\mathcal{L}$, which in this case is the usual linear order on the positive integers. The only requirement on this labeling is that its combination with the discrete order forms a topological sort\footnote{It is possible to deal with lists that have repeated elements. These would need to be modelled by topological sorts for which conditions are relaxed to allow for repeated labels. See \cite{sch1, ear1} for a discussion of how repeated labels can be handled through the assignment of random tie-breakers. It is standard practice in algorithmic analysis to undertake the analysis in first instance for lists \emph{without} duplicate elements--an approach adopted here.}. 

Sorting algorithms hence transform topological sorts of the \emph{discrete order (permutations)} into a unique\footnote{``Unique" in the sense of topological sorts of the linear order using the same labels as the input permutation.} topologic sort of the \emph{linear order (the sorted list)}. The transformation of the list $(9,6,3,2)$ into the sorted list $(2,3,6,9)$ by a sorting algorithm is represented in Figure \ref{Fig:sorting-example}. $(2,3,6,9)$ forms the unique topological sort of the linear order (using the labels 2, 3, 6 and 9 under consideration).

\begin{figure}[h]
\centering
\includegraphics[height=2.7cm]{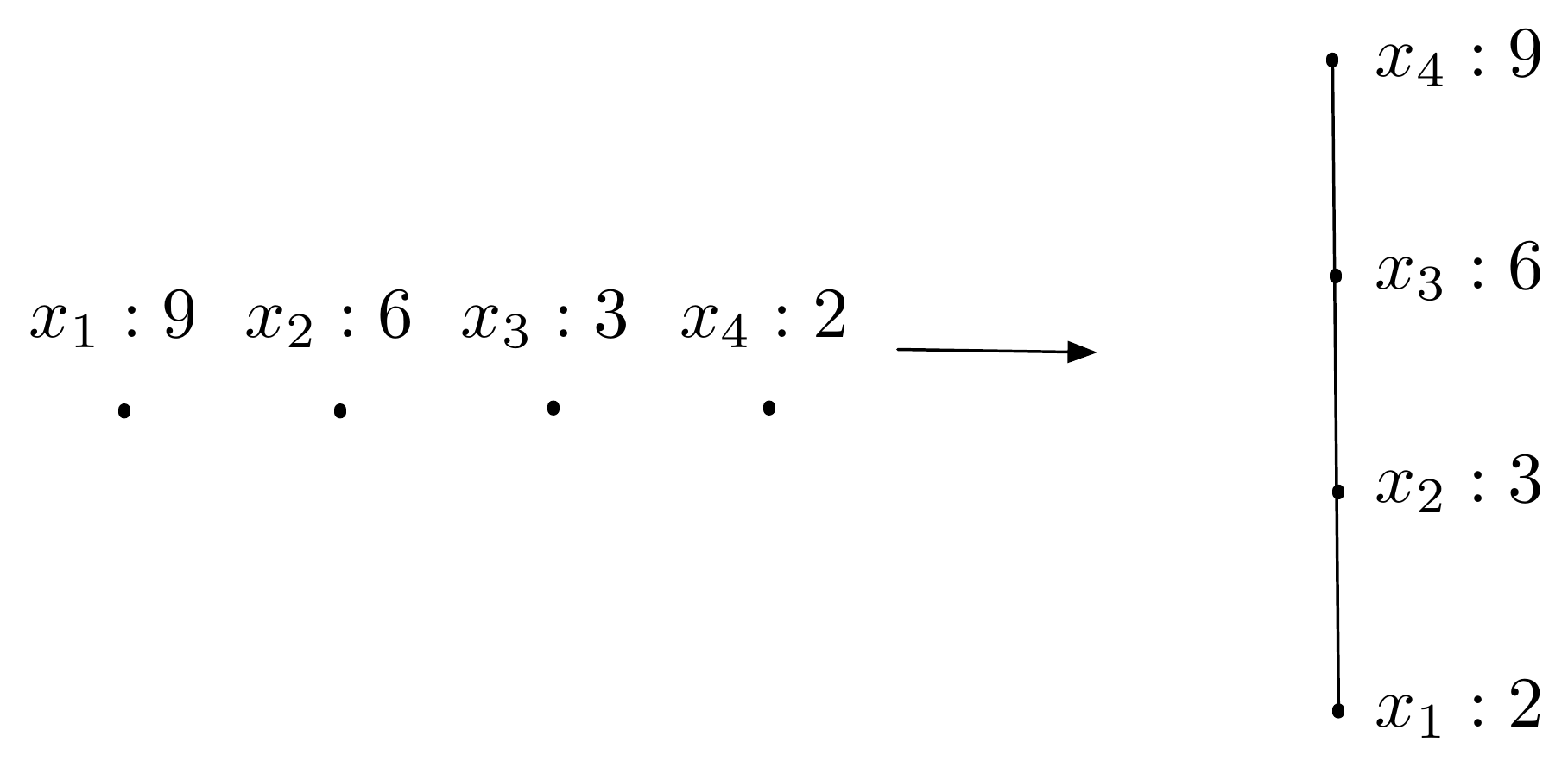}
\caption{Sorting: transforming a topological sort of the \emph{discrete order} into the unique topological sort of the \emph{linear order} \label{Fig:sorting-example}}
\end{figure}
\subsubsection{Global state, part (c)} 
Every list of size n, after identification up to isomorphism with a root state, corresponds to one of the n! permutations of size n. The corresponding state space $R(\Delta_{n}) = \{\sigma_1, \ldots, \sigma_{n!}\}$ consists of the root states, represented as permutations in this case. The multiset $\{(R(\Delta_{n}),1)\}$ containing a single copy of the state space $R(\Delta_{n}) = \{\sigma_1, \ldots, \sigma_{n!}\}$, forms the \emph{(global) state} of the discrete order of size $n$. This global state intuitively represents the possible inputs for the sorting algorithm. 

\subsubsection{Induced transformations on global states, part (d)}
Consider the root states of the discrete order $\Delta_n$ of size $n$, corresponding to $n!$ permutations of size $n$, forming the state space $R(\Delta_{n})$. The (global) state of the discrete order of size $n$ is the multiset $\{(R(\Delta_n),1)\}$, representing the inputs of our algorithm. Every sorting algorithm transforms the root states of this global state into $n!$ copies of the state space of the linear order $L_n$, consisting of a unique root state (a topological sort corresponding to the sorted list).  We obtain the following result. \\

\noindent \textbf{Global state preservation for sorting} \\
Comparison-based sorting algorithms, for inputs of size $n$, transform the global state $\{(R(\Delta_n),1)\}$ into the global state $\{(R(L_n),n!)\}$. 


\subsubsection{Global state preservation: a word of caution} \label{caution1}
We have established the first obvious fact: all comparison-based sorting algorithms preserve global states. Note however that, even though \emph{every comparison-based sorting algorithm} can be naturally interpreted to induce a transformation on global states, this does not entail that \emph{every operation} used in a comparison-based algorithm preserves global states  (cf. \cite{sch1})\footnote{In a sense, it is counter-intuitive that the whole, i.e. a comparison-based sorting algorithm, satisfies the property while some of its operations may not. Global state preservation for all operations is a crucial requirement for feasible modular time analysis: \emph{the analysis of comparison-based algorithms is guaranteed to be feasibly modular, in case \emph{every} operation of the computation preserves global states}  \cite{sch1, sch3}\footnote{In which (global) states are referred to as ``random bags".}. A breakdown of global state preservation for one or more operations lies at the heart of open problems in algorithmic analysis \cite{sch1}. A basic example for which global state preservation breaks down is provided by Heapsort's Selection Phase \cite{sch1}, for which the \emph{exact} time is an open problem \cite{knu}.}. The property of global state preservation can be refined to (global) entropy conservation, motivated in the next section. 

\section{Entropy conservation}
Entropy considerations naturally arise in the context of comparison-based algorithms, e.g. via the well-known $\Omega(log_2(n!))$ lower bound for both the worst-case and average-case time of comparison-based algorithms (on inputs of size $n$) \cite{knu}. For the case of unordered lists of size $n$, the entropy of the input data is $log_2(n!)$. \emph{This notion of entropy will be generalized to the context of topological sorts in a natural way, by measuring the log in base 2 of the number of topological sorts of a finite order.} Our investigation of entropy and its conservation is carried out for computation \emph{with history} (see also \cite{knu}). 

\subsection{Computation with history over topological sorts}  \label{history}
Comparison-based computation typically executes swaps of elements based on comparisons, generalizing the case of comparison-based sorting. In computations \emph{with history}, the original index $\textcolor{red}{i}$ of each label $\textcolor{blue}{a}$ in the input data (topological sort) is paired with the label $\textcolor{blue}{a}$ to form an index-label pair $(\textcolor{red}{i},\textcolor{blue}{a})$. Such a pair replaces each label $a$ in the computation. I.e., instead of exchanging \emph{labels}, a computation with history exchanges \emph{index-label pairs}. Note that the comparisons (that determine the swaps) are still made on the \emph{labels} $\textcolor{blue}{a}$ of a pair  $(\textcolor{red}{i},\textcolor{blue}{a})$. The indices $\textcolor{red}{i}$ are merely carried along for bookkeeping purposes, recording the original position of the label. 

For instance, the trivial sorting example discussed in Example \ref{trivial} that sorts a list of size 2 by a (potential) swap following a single comparison, will leave the topological sort $(x_1: 1,x_2:2)$ unchanged and  transforms $(x_1: 2,x_2:1)$ into $(x_1: 1,x_2:2)$. 

The same computation \emph{with history} uses index-label pairs as a new type of labels of topological sorts. This computation leaves the topological sort $(x_1: (\textcolor{red}{1},\textcolor{blue}{1}),x_2:(\textcolor{red}{2},\textcolor{blue}{2}))$ unchanged and transforms the topological sort $(x_1:(\textcolor{red}{1},\textcolor{blue}{2}),x_2:(\textcolor{red}{2},\textcolor{blue}{1}))$ into the topological sort $\{x_1:(\textcolor{red}{2},\textcolor{blue}{1}),x_2:(\textcolor{red}{1},\textcolor{blue}{2})\}$. This computation is illustrated in Figure \ref{Fig:sorting-history}. Further swaps, on larger input lists, may move these index-label pairs to other positions in the topological sort, but will never change these index-label pairs' values during the computation. 

\begin{figure}[h]
\centering
\includegraphics[height=4cm]{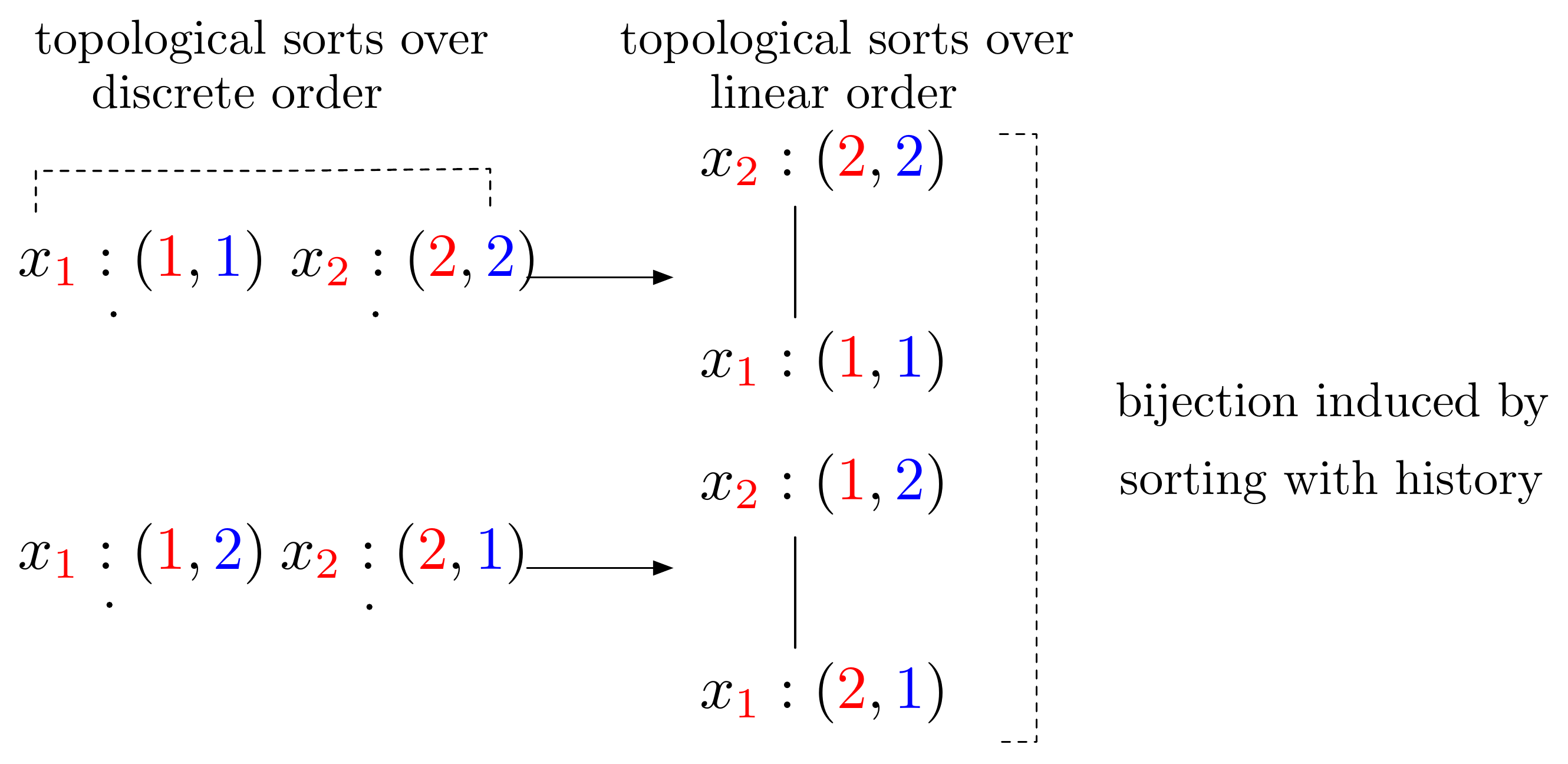}
\caption{Sorting with history for lists of size 2 induces a bijection. Indices (marked red) are paired with labels (marked blue) throughout the computation. Swaps of index-label pairs are based on (blue) label-comparisons only (as was the case for computation without history). (Red) indices are merely carried along in the computation as a bookkeeping device. \label{Fig:sorting-history}}
\end{figure}

Computations with history form a bijection in which the outputs of a computation suffice to determine the inputs. In other words, computations with history are \emph{reversible}, i.e. inputs can be recovered from outputs. In the prior example, the output topological sort $\{x_1:(\textcolor{red}{2},\textcolor{blue}{1}),x_2:(\textcolor{red}{1},\textcolor{blue}{2})\}$ contains the index-label pairs: $(\textcolor{red}{2},\textcolor{blue}{1})$ and $(\textcolor{red}{1},\textcolor{blue}{2})$, which can be ``decoded" to the original input $(x_{\textcolor{red}{1}}: \textcolor{blue}{2},x_{\textcolor{red}{2}}:\textcolor{blue}{1})$. This decoded input, written in history-notation (using index-label pairs instead of original labels), recovers the original input: $(x_1:(\textcolor{red}{1},\textcolor{blue}{2}),x_2:(\textcolor{red}{2},\textcolor{blue}{1}))$. 


\subsection{A basic example of entropy conservation: sorting algorithms} \label{entropy}

 
Comparison-based sorting algorithms compute over $n!$ input lists of size $n$, represented as the root states from the state space over the discrete order of size $n$. As observed, computations with history induce a bijection between inputs and outputs. Viewed over all outputs, the indices of the index-label pairs end up in random order. Indeed, any comparison-based sorting algorithm computing with history and starting from the input permutation \hfill $\sigma = \{(1,\sigma(1)),(2,\sigma(2)),\ldots, (n,\sigma(n))\}$ \\
will produce the sorted output \hfill $ \{(\sigma^{-1}(1),1),(\sigma^{-1}(2),2),\ldots, (\sigma^{-1}(n),n)\}$

We illustrate the transformations on topological sorts induced by a comparison-based sorting algorithm computing with history on all input permutations of size $n$ in Figure \ref{Fig:sorting-history-general}.


\begin{figure}[h]
\centering
\rotatebox[origin=c]{-90}{\includegraphics[height=9cm]{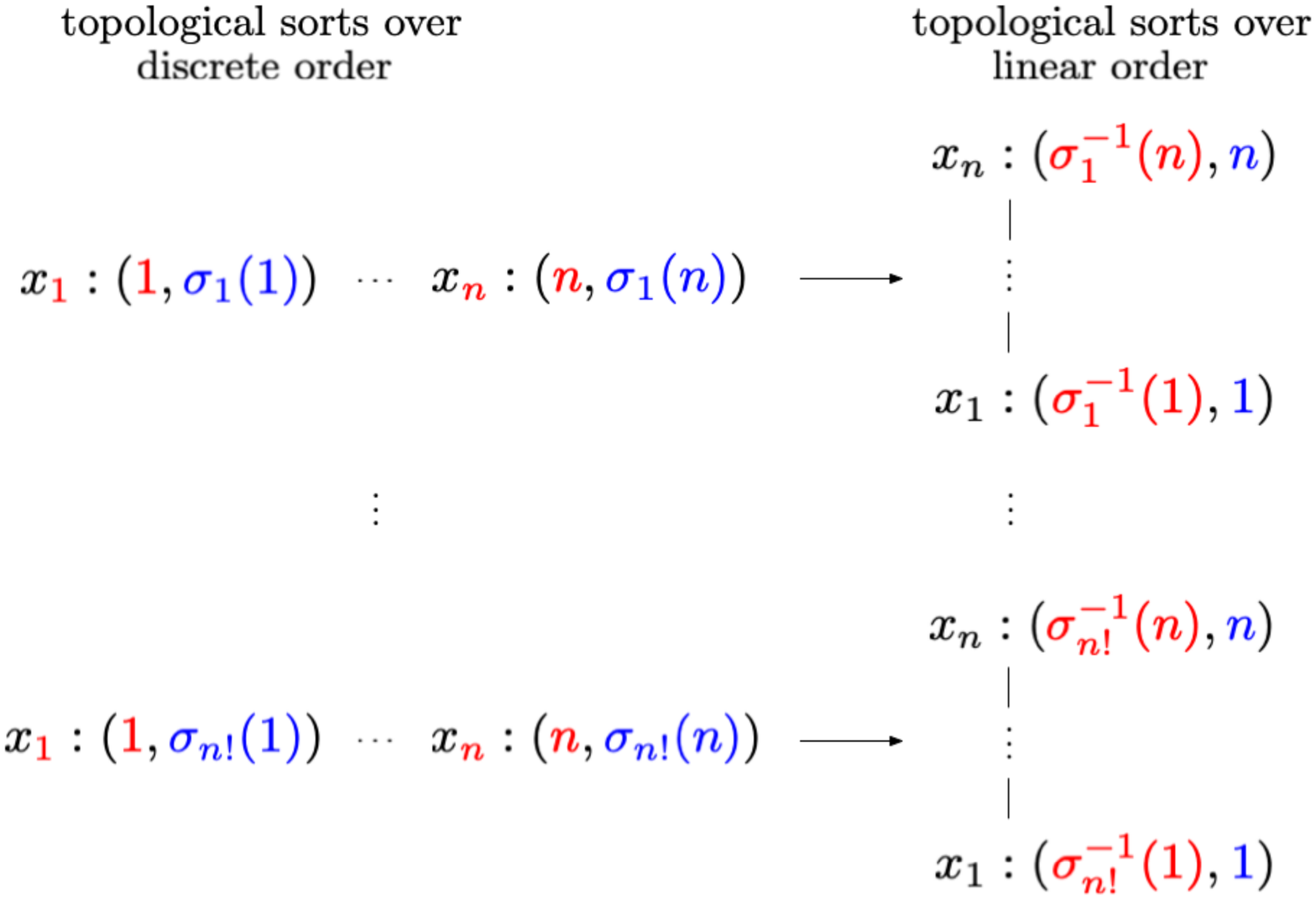}} \vspace*{-1 cm}
\caption{Sorting with history for lists of size n induces a bijection. Indices (marked red) are paired with labels (marked blue) throughout the computation. \label{Fig:sorting-history-general}} 
\vspace*{-4mm} 
\end{figure}

As is clear from Figure \ref{Fig:sorting-history-general}, labels, originally in random order, i.e., uniformly distributed, now occur sorted, i.e. in linear order. Indices, originally in sorted order, linearly arranged from position $1$ to position $n$, after travelling with the labels as index-label pairs during swaps, ultimately occur in a random order, i.e. uniformly distributed. This can be understood by considering that when $\sigma$ ranges over all permutations of size $n$, $\sigma^{-1}$ ranges over the same set of $n!$ permutations. Hence, $\sigma = \{(1,\sigma(1)),(2,\sigma(2)),\ldots, (n,\sigma(n))\}$ when varying over all permutations, causes $ \{(\sigma^{-1}(1),1),(\sigma^{-1}(2),2),\ldots, (\sigma^{-1}(n),n)\}$ to range over all $\sigma^{-1}$, i.e. over all permutations of size $n$. 

At this stage, the linear arrangement of indices is merely an intuitive observation. The indices satisfy the linear order on integers, dictated by the (implicit) ``left-to-right" occurrence of the indices in the original permutation $\sigma = \{(1,\sigma(1)),(2,\sigma(2)),\ldots, (n,\sigma(n))\}$. We will incorporate the linear arrangement of indices explicitly in the context of topological sorts next. 

\subsubsection{Representing the order on indices: left-to-right order} \label{left} 
For traditional list data structures, each input list incorporates a left-to-right order on indices. For instance, for a list $[3,1,2]$, the indices (i.e. positions of the elements) increase in left-to-right order from position $1$ to $3$. This ensures that the list $[3,2,1]$ is not equivalent to, say, the list $[1,3,2]$, i.e. the order of elements is important, contrary to the case of sets. We model lists via topological sorts (root states) of the discrete order. A discrete order is defined on a finite set, in which the order of elements plays no role. To ensure that the discrete order faithfully models the list data structure, we need to impose one more condition on the order $\Delta_n$ over a finite set $X$. 

If $X$ is enumerated as $X = \{x_1,\ldots,x_n\}$, then the order on indexed elements $x_i$ can be imposed via a ``left-to-right order''\footnote{The left-to-right order was first introduced on orders in \cite{ear1}.}, i.e., we impose a second order on the elements $x_i$, in addition to the discrete order. These orders are distinct, i.e. do not affect one another. The left-to-right order $\sqsubseteq^{*}$ specifies that $x_i \sqsubseteq^{*} x_j$ if an only if $i < j$, i.e. the linear order on integers is inherited on the elements of the finite set, forming a linear order $\sqsubseteq^{*}$, denoted in our context by $L_n$, on $X$.\footnote{Note that this definition is specific to discrete orders. We will generalize the left-to-right order to more general partial orders later.}

Once the left-to-right order is imposed on the elements of $X$, indices can be interpreted to form the single root state of a linear order of size $n$. 

Note that due to Hasse diagram representation, the ``left-to-right" order will be represented ``vertically" in the Hasse diagram of a linear order rather than ``horizontally" (as in the traditional list data structure format). Hence we abandon the terminology ``left-to-right" in favour of that of a ``dual order'' in section \ref{mirror}.

\subsubsection{Splitting the roles of indices and labels in topological sorts} \label{splitting}
With the introduction of the left-to-right order, indices are interpreted as topological sorts of the linear order. Hence the computation displayed in Figure \ref{Fig:sorting-history-general} can be split up, where we consider the first coordinate $i$ of a label-index pair $(i,a)$ \emph{separately} from the label-coordinate $a$. This yields a representation using two state spaces, one for labels over a discrete order, the other for indices over a linear order displayed in Figure 6.

\begin{figure}[!ht]
    \centering
    \caption{Separating the state spaces}
        \subfloat[Case of {\bf labels}]{\includegraphics[width=0.4\columnwidth]{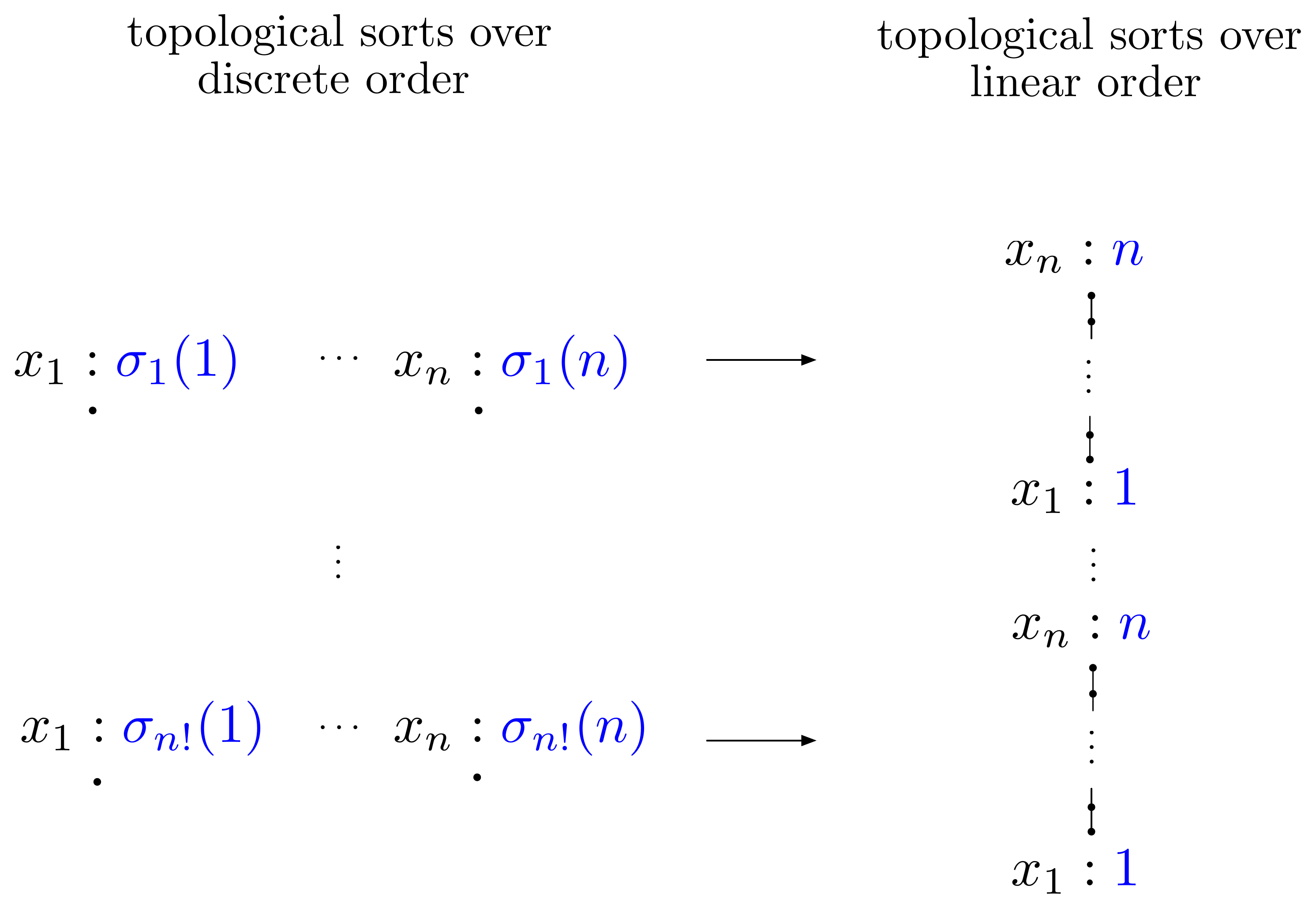}}
        \qquad \qquad \qquad
        \subfloat[Case of {\bf indices}]{\includegraphics[width=0.4\columnwidth]{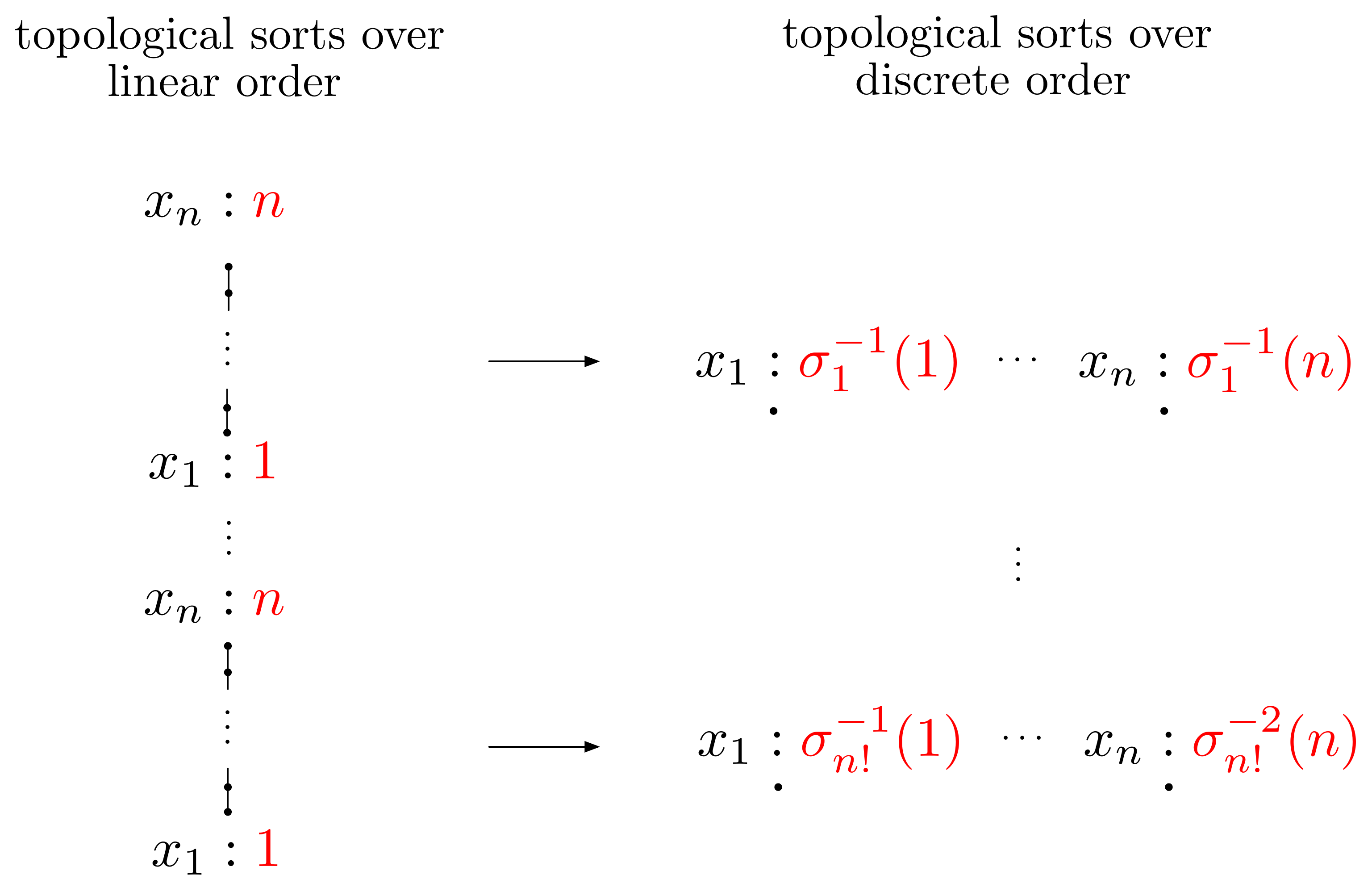}}
    \label{fig:subfigname}
\end{figure}

\subsubsection{Dual order} \label{mirror} 
The orders involved for labels and for indices (for the case of comparison-based sorting algorithms) are \emph{dual} in the following sense: the Hasse diagram of the discrete order $\Delta_n$ represents the \emph{complement} graph of the Hasse diagram of the linear order $L_n$. This type of duality plays a central role in our entropy conservation results and we will introduce the notion of a dual order for more general types of partial orders in Section \ref{dualSP}. 


There is a natural transformation to obtain the dual of the discrete order, i.e. the linear order. One can be obtained from the other by taking the mirror image in the Cartesian plane, where points of the order are represented via a choice of coordinates placing the elements of the discrete order on, say, a horizontal line. The mirror image occurs with respect to the first diagonal. The dual order is formed by the complement graph on the mirror image, forming the Hasse diagram of the linear order. We illustrate this in Figure \ref{Fig:mirror-discrete-linear}.

\begin{figure}[h]
\centering
\includegraphics[height=3.5
cm]{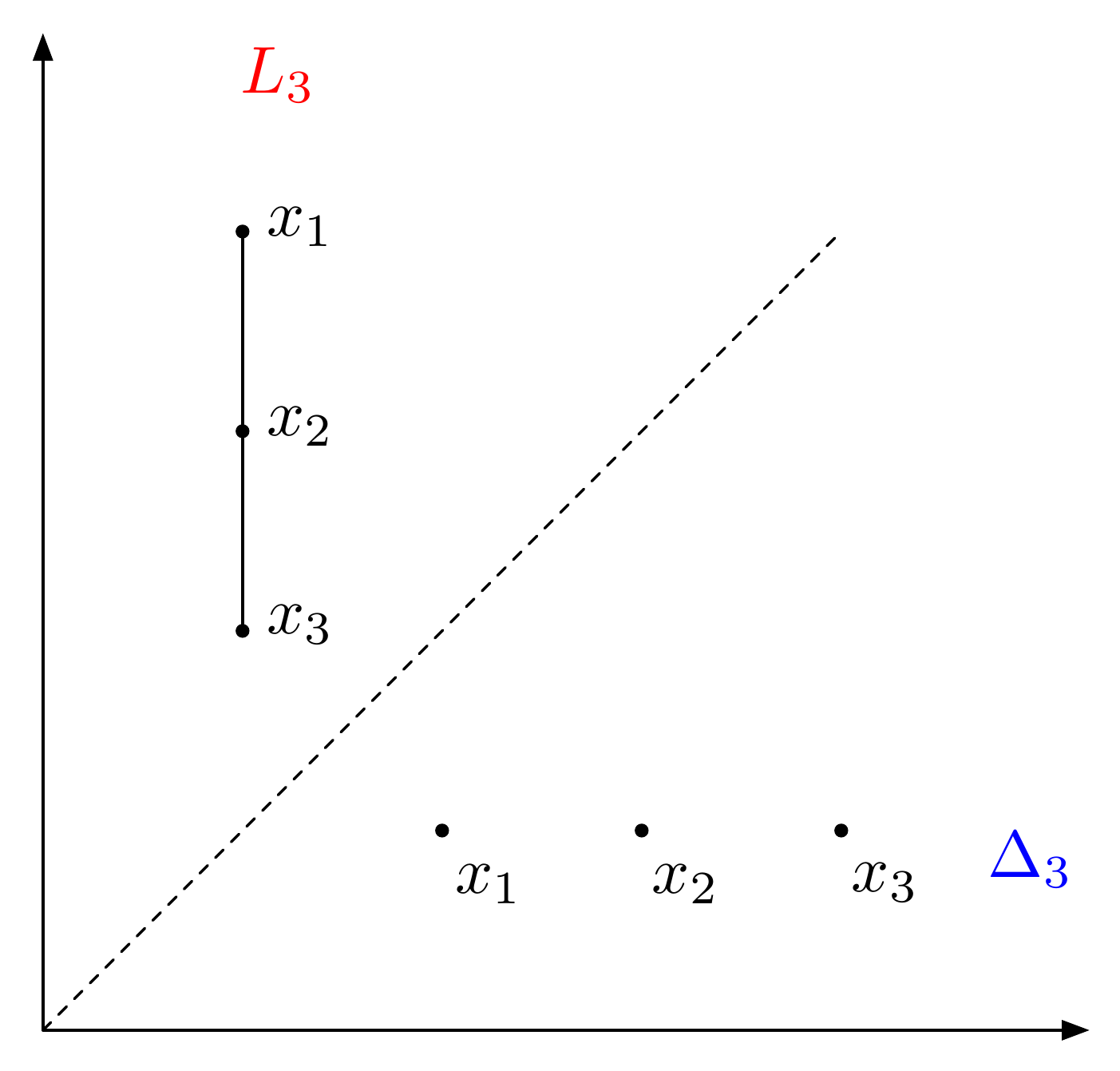}
\caption{Forming the dual order via the mirror image wrt the first bisector \label{Fig:mirror-discrete-linear}}
\end{figure}


\subsubsection{Entropy conservation for comparison-based sorting} \label{entropy-topsort}
Next, we pair the \emph{topological sorts}, i.e. instead of considering the original index-label pairs used in the computation with history, we consider \emph{pairs of topological sorts}, one over a discrete order (representing the order on labels) and one over its dual, a linear order (representing the order on indices) as illustrated in Figure \ref{Fig:sorting-history-general-paired}. These pairs are transformed by a sorting algorithm into new pairs of topological sorts, one over a linear order (representing the labels) and one over a discrete order (representing the indices).

\begin{figure}[h]
\centering
\includegraphics[height=8
cm]{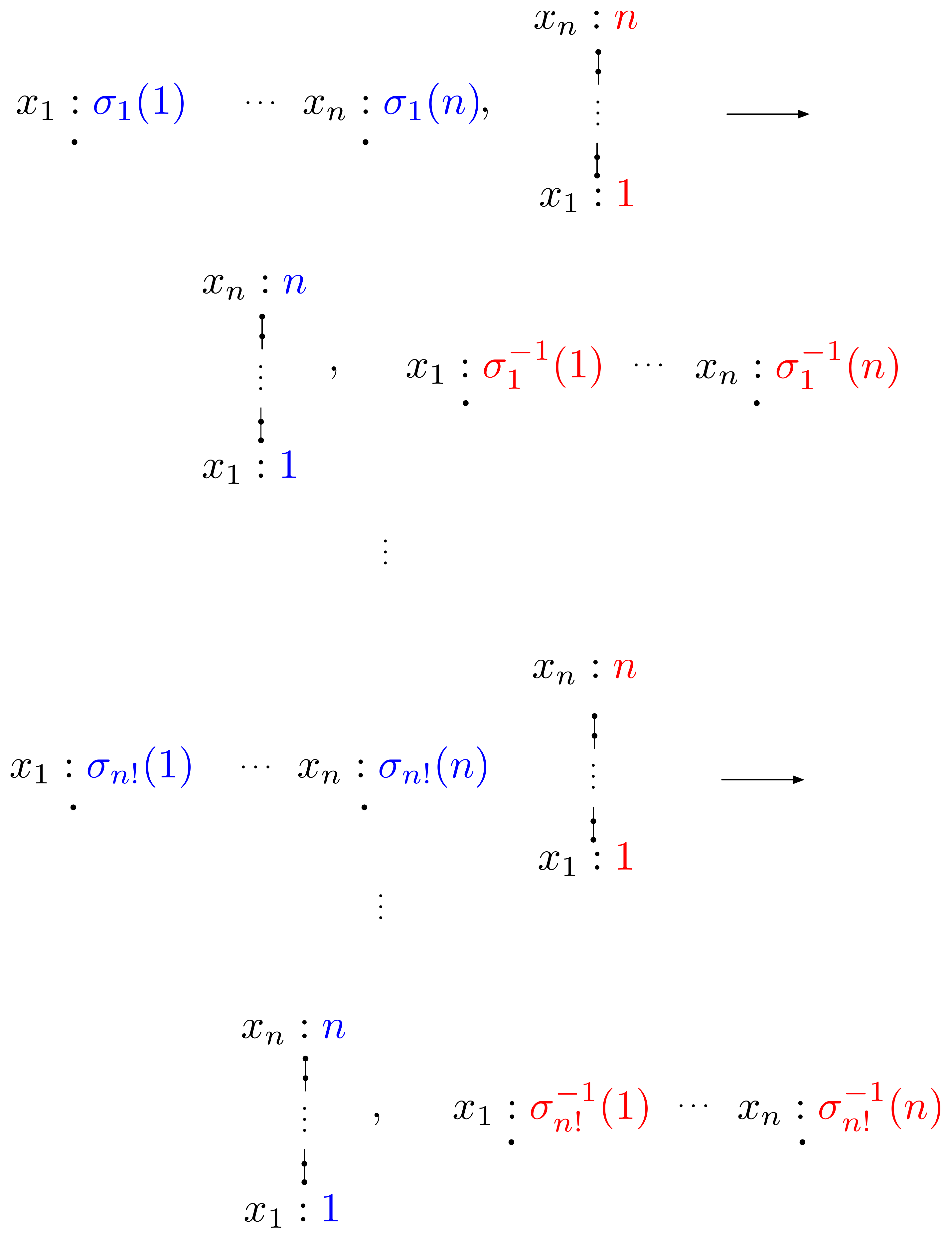}
\caption{Pairing the topological sorts for labels and indices. \label{Fig:sorting-history-general-paired}}
\end{figure}

For the case of labels: the order is transformed from a discrete order (i.e. full freedom on elements) to a linear order (no freedom on elements). Hence randomness at the label level decreases during the computation with history. This is compensated with an increase in randomness at the index level during the same computation, where the order for the case of indices is transformed from a linear order (no freedom on elements) to a discrete order (full freedom on elements). 

Randomness is measured via entropy and we recall that entropy is defined as the log in base 2 of the size of a state space. The above observation on decrease of randomness for labels and increase of randomness for indices can be expressed formally via entropy on inputs and outputs: 

\begin{itemize}
\small{\item{Entropy of {\bf input}-labels for the state space $R(\Delta_n)$: $log_2|R(\Delta_n)| = log_2(n!)$ \\(entropy for the $n!$ states of the state space over the discrete order)}
\item{Entropy of {\bf input}-indices for the \emph{dual} state space $R(L_n)$: $log_2|R(L_n)| = log_2(1) = 0$ \\(entropy for the single state satisfying the linear order)}
\item{Entropy of {\bf output}-labels for the state space $R(L_n)$: $log_2|R(L_n)| = log_2(1) = 0$ \\(entropy for  the single state of the output linear order)}
\item{Entropy of {\bf output}-indices for the \emph{dual} state space $R(\Delta_n)$: $log_2|R(\Delta_n)| = log_2(n!)$ \\(entropy for $n!$ states of the discrete order)} }
\end{itemize} 

\begin{definition}
Consider the discrete order $\Delta_n$ and its dual order $L_n$. The corresponding \textbf{double state space} is the pair $(R(\Delta_n),R(L_n))$, consisting of a state space and its dual state space.

The \textbf{quantitative entropy} (``label entropy''), denoted by $H_q$, is the entropy of the first component of the double state space: $H_q = log_2|R(\Delta_n)| = log_2(n!)$.

The \textbf{positional entropy} (``index entropy''), denoted by $H_p$, is the entropy of the second component of the double state space: $H_p = log_2|R(L_n)|$

\end{definition}

The quantitative entropy $H_q$ for the case of a comparison-based sorting algorithm evolves from maximum entropy $log_2(n!)$ to 0 entropy. 
The positional entropy $H_p$ for the case of a comparison-based sorting algorithm evolves from 0 entropy to maximum entropy $log_2(n!)$.

In other words, comparison-based sorting algorithms satisfy entropy conservation, i.e. the sum $H_p + H_q$ of quantitative and positional entropy remains constant, $H_p + H_q = log_2(n!)$, for the case of inputs \emph{and} for the case of outputs. In summary, we obtain a refined version of global state preservation, involving pairs of global states and a related entropy conservation result.

\begin{proposition} (Entropy conservation for comparison-based sorting) \label{entropyconservationsorting} \\
 Comparison-based sorting algorithms transform the double global state: $(\{(R(\Delta_n),1)\}, \{(R(L_n),n!)\})$ into the double global state \hfill $(\{(R(L_n),n!)\}, \{(R(\Delta_n),1)\})$ 

Let the entropies for the double state space $(R(\Delta_n),R(L_n))$ be $$H^{1}_{p} = log_2(|R(L_n)|) \mbox{ and } H^{1}_{q} = log_2(|R(\Delta_n)|)$$ Similarly, let the entropies for the double state space $(R(L_n),R(\Delta_n))$ be $$H^{2}_{p} = log_2(|R(\Delta_n)|) \mbox{ and } H^{2}_{q} = log_2(|R(L_n)|)$$ The sum of the entropies of the states spaces involved in the double global states is constant $$H^{1}_p + H^{1}_q = H^{2}_p + H^{2}_q = log_2(n!)$$

In other words, \emph{comparison-based sorting satisfies entropy conservation}. 
\end{proposition}

\subsubsection{Entropy conservation: a word of caution}  \label{caution2}
A word of caution has its place here too, as it did in Section \ref{caution1} on global state preservation: though \emph{every} comparison-based algorithm can be naturally shown to satisfy entropy conservation, this does not entail that \emph{every} operation used in the algorithm conserves entropy. 

We will show that for the case of SP-orders, to be introduced next, entropy conservation (Proposition \ref{entropyconservationsorting}) has a natural generalization in Theorem \ref{duality}. 

\section{Series-parallel orders} \label{SP}
Series-parallel orders, or SP-orders, form an important, computationally tractable class of data structures. These include trees and play a role in sorting, sequencing, and scheduling
applications \cite{moh}. \cite{sch1} introduces a calculus supporting the modular time derivation of algorithms, where the suite of data structuring operations used in \cite{sch1} preserves SP-orders. SP-orders are generated from a finite set of elements using a ``series" and ``parallel" operation over partial orders. 

\subsection{Series-parallel operations} \label{operations}
The series and parallel operations, or SP-operations, $\otimes$ and $\parallel$ are defined over partial orders.

 In terms of Hasse diagrams:

\begin{itemize}
\item{the series operation puts the first poset below the second, where every element of the first order ends up (in the newly formed order) below each element of the second order.} 

\item{the parallel operation puts the two orders side-by-side, leaving their nodes mutually incomparable (across the two orders).} 
\end{itemize} 

Figure \ref{Fig: series-parallel-operations} illustrates the series operation on orders with $\vee$-shaped and $\wedge$-shaped Hasse diagrams, where we introduce the formal definitions of the operations next.

\begin{figure}[h]
\centering
\includegraphics[height=6cm]{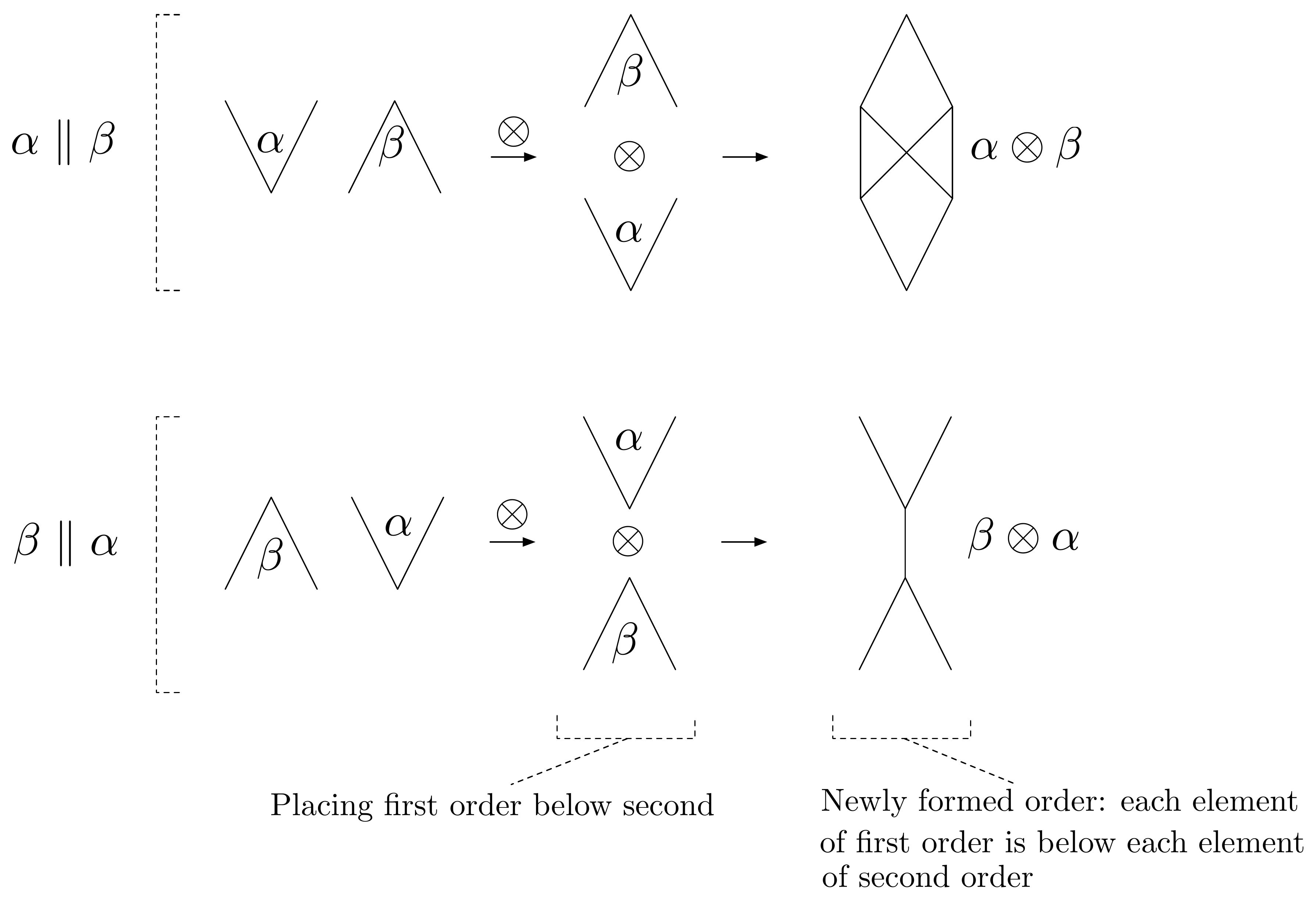}
\caption{The series operation execution illustrated on Hasse diagrams \label{Fig: series-parallel-operations}}
\end{figure}

\begin{definition} \label{sp-order}
Let $\alpha_1,\alpha_2$ be posets where $\alpha_1 = (A_1,\sqsubseteq_1)$, $\alpha_2 = (A_2,\sqsubseteq_2)$ and $A_1 \cap A_2 = \emptyset$, then the series composition $\alpha_1 \otimes \alpha_2$ and the parallel composition $\alpha_1 \parallel \alpha_2$ of these orders are defined to be the following partial orders

\begin{itemize} 
\item{$\alpha_1 \otimes \alpha_2  \,=\,\,(A_1 \cup A_2,   \sqsubseteq_1  \cup \, \sqsubseteq_2 \cup \, (A_1 \times A_2)) $}
\item{$\alpha_1 \parallel \alpha_2 \,=\,\, (A_1 \cup A_2,   \sqsubseteq_1  \cup \, \sqsubseteq_2)$}
\end{itemize}

If $\alpha=\alpha_1\otimes\alpha_2$, then $\alpha_1$ and $\alpha_2$ are \textbf{series components} of $\alpha$. 

If $\alpha=\alpha_1\|\alpha_2$ then $\alpha_1$ and $\alpha_2$ are \textbf{parallel components} of $\alpha$. 
\end{definition}



\begin{remark} \label{dualordersremark} It is easy to verify that both the series and the parallel operation over partial orders are associative. The series operation is clearly not commutative, while the parallel operation is. 

Note however that for the purpose of setting up the computational model for modular timing, in which partial orders model data structures, we will require in practice that the parallel operation is \underline{not} commutative. This will be motivated and formalized in Section \ref{dualSP}. 

\end{remark}

SP-orders can be characterized as the N-free partial orders \cite{moh}, i.e. the orders which do not contain the``N-shaped partial order'' as sub-order. An N-shaped partial order is any order determined by a quadruple $\{x,y,u,v\}$ for which $x \sqsubseteq y, u \sqsubseteq v,  x \sqsubseteq v$ and $u$ and $y$ are unrelated. 


SP-orders include partial orders with tree-shaped Hasse-diagrams, i.e. capture tree data structures. Of course, SP orders need not be tree-shaped, as illustrated by the left-most example in Figure \ref{Fig: SP-orders}. The figure on the right-hand side is an example of a non-SP-order containing the N-shaped sub-order. 

\begin{figure}[h]
\centering
\includegraphics[height=3cm]{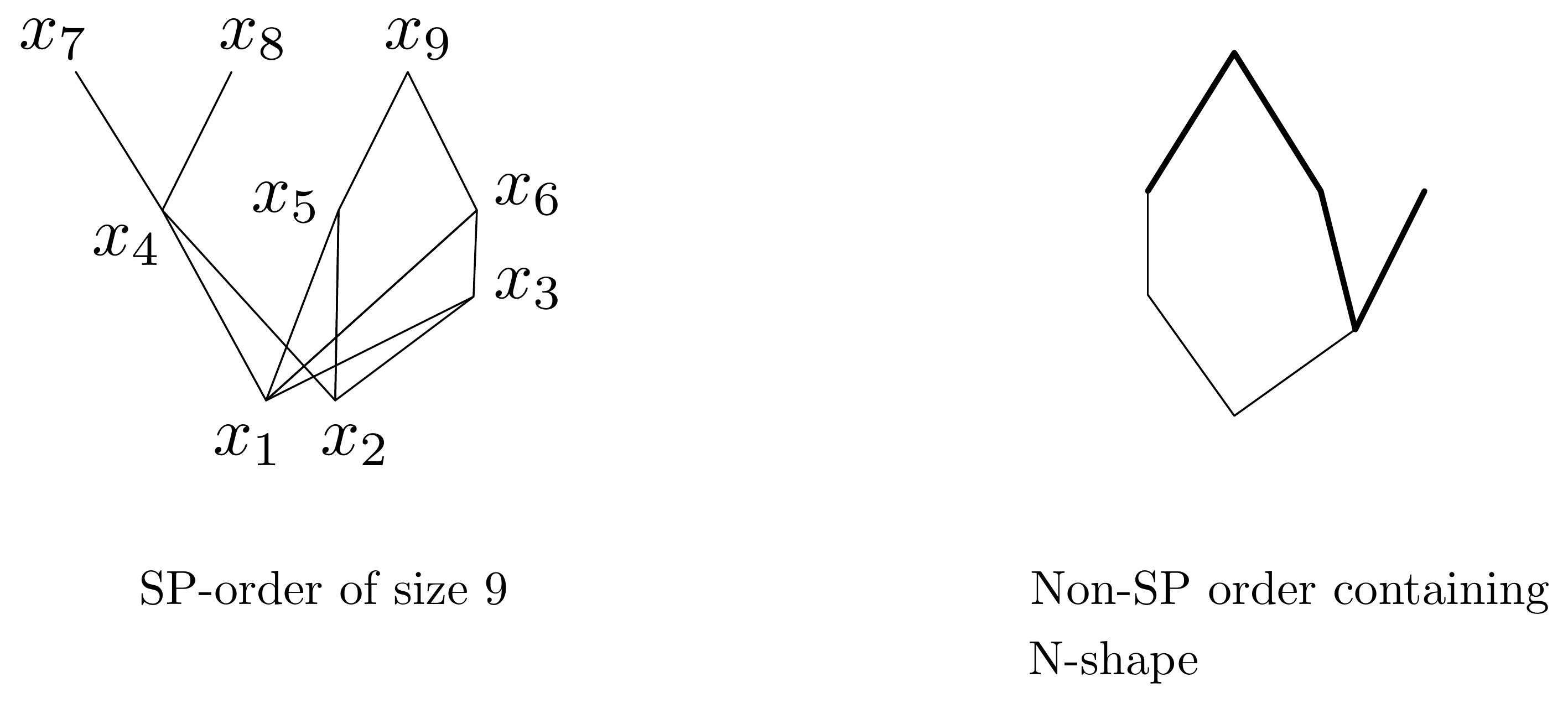}
\caption{SP and non-SP order \label{Fig: SP-orders}}
\end{figure}

\subsection{The series operation is refining}
Comparison-based computations gather information by comparing elements. Order-information gained in this way is represented via partial orders in our context. As the computations proceed, more order-information is obtained. This is captured by the notion of a refinement of the order. 

\begin{definition} \label{refinement} 
A poset $\beta = (X,\sqsubseteq_{\beta})$ \textbf{refines} a poset $\alpha = (X, \sqsubseteq_{\alpha})$, denoted by $\alpha \preccurlyeq \beta$, in case there is a permutation $\sigma$ on (the indices of) elements of $X = \{x_1,\ldots,x_n\}$ such that  for all $x_i, x_j \in X$ $$x_i \sqsubseteq_{\alpha}  x_ j \Rightarrow x_{\sigma(i)} \sqsubseteq_{\beta} x_{\sigma(j)}$$ 

I.e., $\beta$ refines $\alpha$ exactly (in case they have the same domain) when $\alpha$ is a sub-order of $\beta$\footnote{Or, with different domains, when there is an order-preserving map from one into the other.}.
\end{definition}

Each application of a series operation refines the order under consideration (cf. Figure \ref{Fig:SP-notation}). Computations over orders in our context typically start from the discrete order (first Hasse diagram in Figure  \ref{Fig:SP-notation}) and gradually build refinements by repeated application of the series operation, e.g. in the construction the tree-structured Hasse diagram (second Hasse diagram in Figure \ref{Fig:SP-notation}) or in a completed sort, resulting in the linear order (third Hasse diagram in Figure \ref{Fig:SP-notation}).

\subsection{Counting the root states of an SP-order} \label{count}
For general partial orders, determining $|R(\alpha)|$ is $\# P$-complete \cite{bw}. The following lemma \cite{moh} reflects the computational tractability of SP-orders, showing how to quickly compute the size of the number of (non-isomorphic) topological sorts of a SP-order. This determines the size $|R(\alpha)|$ of a state space of a SP-order $\alpha$.
\begin{lemma} \emph{{\bf (State space size)}}

\label{Le:|R|}

Let $\alpha_1$ and $\alpha_2$ be posets.
\begin{enumerate}
\item
\label{Le:|R|Ser}
If $\alpha=\alpha_1\otimes \alpha_2$ then $|R(\alpha)|=|R(\alpha_1)|  |R(\alpha_2)|$
\item
\label{Le:|R|Par}
If $\alpha=\alpha_1\|\alpha_2$ then $|R(\alpha)|={|\alpha|\choose |\alpha_1|}|R(\alpha_1)|  |R(\alpha_2)|$
\end{enumerate}
\end{lemma}

\subsection{SP-expressions}
SP-orders are of interest since they determine a computationally tractable class of data structures. Many problems are NP-hard in general. Considerable attention has been given to partial orders with ``nice" structural properties supporting the design of efficient methods \cite{moh}. This includes the class of SP-orders, for which the number of topological sorts can be efficiently computed, while the case for general partial orders, as mentioned, is $\# P-$complete \cite{bw}. 

Intuitively, SP-orders are partial orders created from a finite set $X = \{x_1,\ldots,x_n\}$ equipped with the discrete order, by repeated applications of the series and parallel operation, starting from singleton discrete orders. This is formalized via the notion of a SP-expression and an SP-order determined by a SP-expression. Examples are provided following this formal definition. 


\begin{definition}(general SP-expressions)
\begin{itemize}
\item{A variable $x_i$ is a general SP-expression} 

\item{$\Psi = (\Psi_1 \otimes \Psi_2)$ is a general SP-expression when $\Psi_1$ and $\Psi_2$ are general SP-expressions }

\item{$\Psi = (\Psi_1 \parallel \Psi_2)$ is a general SP-expression when $\Psi_1$ and $\Psi_2$ are general SP-expressions}

\end{itemize}

A \textbf{SP-expression} is a general SP-expression that contains each of its variables $x_i$ at most once. SP-expressions are logical formulae built form the symbolic notations $\otimes$ and $\parallel$ for the series and parallel operation. We denote SP-expressions in the following by Greek letters $\Psi,\Phi$ etc. 

$Var(\Psi)$ denotes the set of variables in $\Psi$.
\end{definition}

\begin{example} The expressions $(x_1 \otimes x_2)$, $(x_1 \otimes x_3) \parallel x_2$ and $x_3 \parallel (x_1 \parallel x_2)$ are SP-expressions. The following is a general SP-expression that is not a SP-expression: $(x_1 \otimes x_2) \parallel x_2$. 
\end{example}

\subsection{SP-orders determined by SP-expressions} \label{order-generation}

An SP-expression $\Psi$ with $n$ variables ``generates" a SP-order of size $n$ starting from the discrete order $(X,\Delta_n)$ where $X$ consists of the set of elements $Var(\Psi) = \{x_1, \ldots, x_n\}$. The generation process is defined as follows: 

\begin{enumerate}
\small{\item{Interpret each variable $x_i$ in $\Psi$ as a singleton discrete order obtained by the restriction of the discrete order $(X,\Delta_n)$ to the element $x_i$} 

\item{Interpret each $\otimes$-symbol in $\Psi$ as a series operation over partial orders}

\item{Interpret each $\parallel$-symbol in $\Psi$ as a parallel operation over partial orders} 

\item{Execute the operations of $\Psi$ over partial orders (as indicated in items 2 and 3 above) in the precedence-order determined by the brackets of the SP-expression $\Psi$} }

\end{enumerate} 

The final result is referred to as \textbf{the SP-order determined by $\Psi$}. 

\begin{remark} It is clear the each SP-order of a given size $n$ can be obtained from a suitably chosen SP-expression wth $n$ variables and a discrete order of size $n$ via the process sketched above. 
 \end{remark}

\begin{example} 
The 4-element order with tree-shaped Hasse diagram in Figure \ref{Fig: topological-sort-example-alt} can be determined by the SP-expression $((x_1 \otimes x_2) \parallel x_3) \otimes x_4$ via the process described above. The result is displayed via the second Hasse diagram in Figure \ref{Fig:SP-notation} which illustrates two refinements of a discrete order $\Delta_4$, resulting in the linear order $L_4$. In each case, SP-expressions are displayed that generate the SP-orders involved.  

\end{example} 

\begin{remark} Given a SP-order determined by SP-expression $\Psi$, then with abuse of terminology we will refer to this order as ``the order $\Psi$'' for the sake of brevity.  For instance, for SP-expression $\Psi  = ((x_1 \otimes x_2) \parallel x_3) \otimes x_4$, we refer to the order generated by this SP-expression as ``the order $((x_1 \otimes x_2) \parallel x_3) \otimes x_4$'', rather than the order ``generated by this SP-expression''. 

We extend this convention also to the context of state spaces and refer to the state space $R(\alpha)$, where $\alpha$ is the order $((x_1 \otimes x_2) \parallel x_3) \otimes x_4$, as the state space $R(((x_1 \otimes x_2) \parallel x_3) \otimes x_4)$.  \end{remark}

\begin{figure}[h]
\centering
\includegraphics[height=2.7cm]{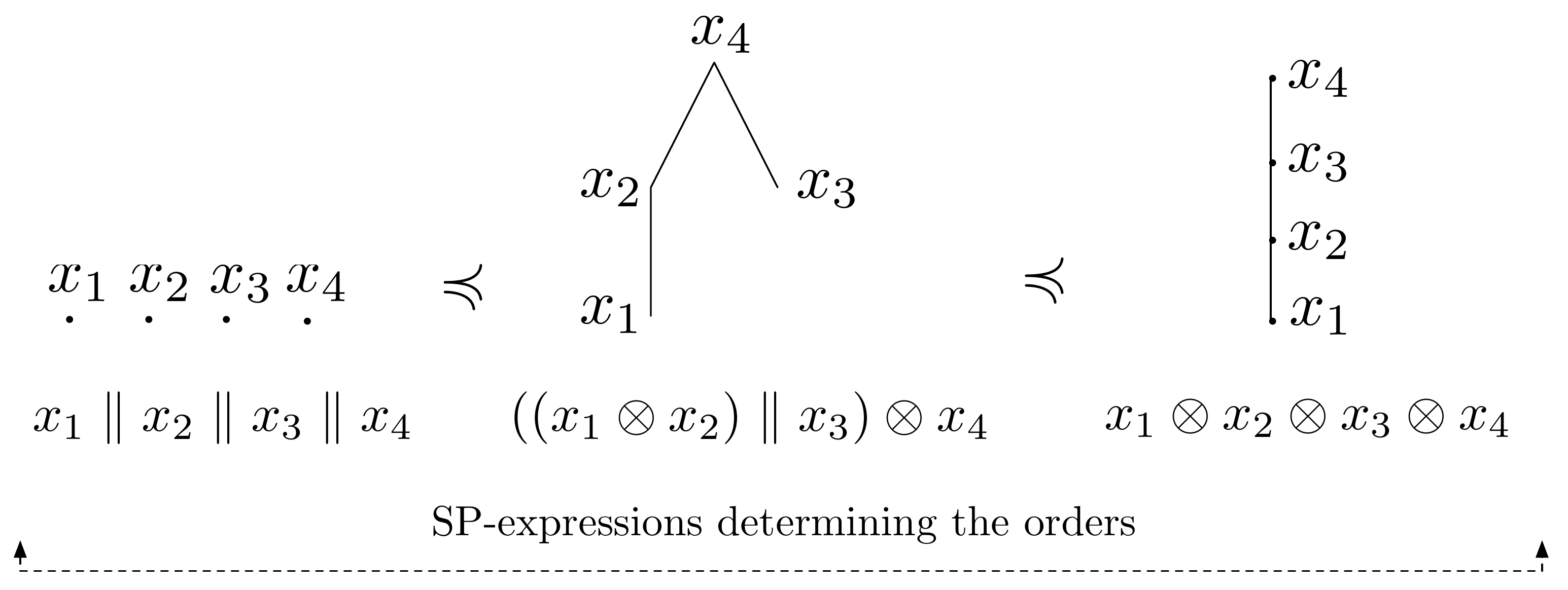}
\caption{refinements of the discrete order and their respective SP-expressions \label{Fig:SP-notation}}
\end{figure}

Finally we include an application of Lemma \ref{Le:|R|} using the SP-expression notation.

\begin{example}The state space discussed in Figure \ref{Fig: topological-sort-example-alt} contains three root states: V, VI and VII, representing the heaps of size $4$. A corresponding SP-order $((x_1 \otimes x_2) \parallel x_3) \otimes x_4$ over the four-element set $X = \{x_1,x_2,x_3,x_4\}$ is displayed in Figure \ref{Fig:SP-notation}. Applying Lemma \ref{Le:|R|}, we obtain: $|R((x_1 \otimes x_2) \parallel x_3) \otimes x_4| = (({3\choose 2} \times (1 \times 1)) \times 1) \times1 = 3$, corresponding to the 3 topological sorts (V, VI and VII) of this order displayed in Figure \ref{Fig: topological-sort-example-alt}. 
\end{example}

\section{Dual SP-orders} \label{dualSP}
To ensure that SP-orders faithfully model data structures, we need to impose an additional partial order on SP-orders, leading to the notion of ``a dual SP-order''\footnote{Originally introduced in \cite{ear1}, be it in the left-to-right order context--not a dual order context.}. 

\textbf{\emph{In our context, in which we model data structures, the parallel operation in SP-orders will be adapted to a non-commutative version}}. As was the case for our discussion of the left-to-right order for the particular case of the discrete order in Section \ref{left}, this will be achieved by imposing an additional order. The additional order serves to distinguish ``left" and ``right" sub-orders $\alpha_1$ and $\alpha_2$ respectively of a parallel composition $\alpha_1 \parallel \alpha_2$. 

We recall the motivation for this additional order, discussed in Sections \ref{left}, \ref{mirror} and Remark \ref{dualordersremark} for the particular case of discrete orders, where we restate matters using the SP-notation. 

The discrete order $x_1 \parallel x_2  \parallel \ldots  \parallel x_n$ represents an unordered \emph{list} $L = [x_1,\ldots,x_n]$. For regular data structures such as lists or arrays, the order of the indices is not interchangeable, i.e., $L = [x_1,x_2]$ is not equivalent to $L' = [x_2,x_1]$ and hence $x_1 \parallel x_2$ should not be equivalent to $x_2 \parallel x_1$. Similarly, for a heap of size 4, the data structure in our model is represented by the SP-order $((x_1 \otimes x_2) \parallel x_3) \otimes x_4$. This SP-order should not be equivalent to the SP-order $(x_3 
\parallel (x_1 \otimes x_2)) \otimes x_4$, since a heap-data structure corresponds to a complete binary tree in which some leaves may have been removed in right to left order only \cite{knu}. 

Imposing an additional order ensuring non-commutativity of the parallel operation also matters in terms of the state space. With a commutative parallel operation, the state space $R(\Delta_n) = R(x_1 \parallel x_2  \parallel \ldots  \parallel x_n)$ consists of the $n!$ permutation-topological sorts. For a commutative parallel operation the state space of $R(x_1 \parallel x_2  \parallel \ldots  \parallel x_n)$ would reduce to a \emph{single} topological sort. 

For the case of the discrete order $x_1 \parallel x_2  \parallel \ldots  \parallel x_n$ we can impose a left-to-right order $\sqsubseteq^{*}$ on the indexed elements $x_i$ by requiring that $x_i \sqsubseteq^{*} x_j \Leftrightarrow i < j$. Hence the imposed left-to-right order $\sqsubseteq^{*}$ is the linear order $(X,L_n)$. 

Using the notation for SP-orders, the original order is the discrete order $x_1 \parallel x_2  \parallel \ldots  \parallel x_n$ while the newly imposed order is the linear order $x_1 \otimes x_2  \otimes \ldots  \otimes x_n$. This is exactly the ``dual order" obtained from the original one by interchanging parallel operations with series operations and vice versa. Just as for the case of the series operation, for which a series composition $\alpha_1 \otimes \alpha_2$ places all elements of $\alpha_1$ below each element of $\alpha_2$ in the given order $(X,\sqsubseteq)$, we impose a second order ``left-to-right order'' $(X,\sqsubseteq^{*})$ ensuring that \emph{under this order} a parallel composition $\alpha_1 \parallel \alpha_2$ will place every element of $\alpha_1$ below each element of $\alpha_2$. I.e. the two orders are duals, switching the roles of the series and parallel operations. This naturally leads to the notion of a ``dual SP-order''. 

\subsection{Dual SP-order and double state space}
\begin{definition}  The \textbf{dual of a SP-expression} $\Psi$ is the SP-expression $\Psi^{*}$ obtained by interchanging series operations and parallel operations. The dual of the SP-expression $((x_1 \otimes x_2) \parallel x_3) \otimes x_4$ is the SP-expression $((x_1 \parallel x_2) \otimes x_3) \parallel x_4$. 

Given an SP-order $\alpha$ determined by an SP-expression $\Psi$, then the \textbf{dual SP-order} $\alpha^{*}$ is the SP-order determined by the dual SP-expression $\Psi^{*}$.


Given SP-order $\alpha$ and its dual $\alpha^{*}$, then $R(\alpha^{*})$ is a \textbf{dual state space} of the state space $R(\alpha)$ and $(R(\alpha),R(\alpha^{*}))$ is a \textbf{double state space}.


\end{definition}

We recall from Section \ref{mirror} that the Hasse diagram representations of the discrete order and its dual could be obtained via a reflection along the first bisector in the Cartesian plane. This is also the case for a general SP-order and its dual.

\subsection{Cartesian construction of the dual order} \label{quadrant}

We generalize the Cartesian construction of the Hasse diagram of the dual of a discrete order to general \emph{SP-orders}.

\begin{center} \includegraphics[height=1.1cm]{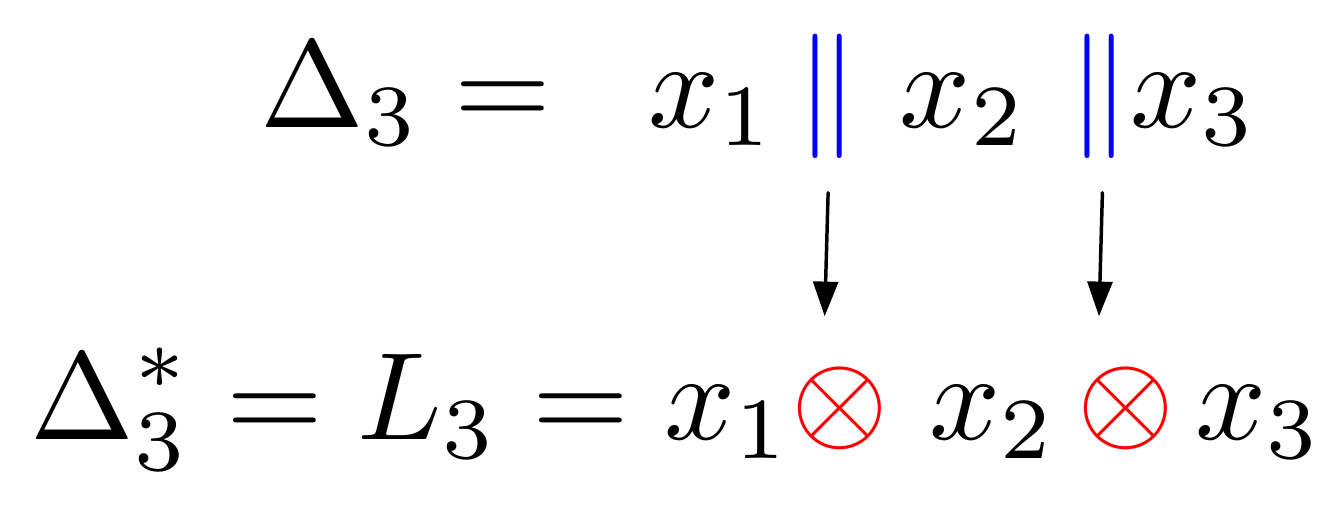} \end{center}

\begin{center}

\includegraphics[height=3.5cm]{mirror-discrete-linear} \end{center}

\noindent {\bf Dual SP-order construction} \\
Given an SP order $\alpha$. Consider the Cartesian plane and draw the Hasse diagram of the order $\alpha$ in the lower half of the bisected first quadrant (for a suitable choice of coordinates of the Hasse diagram picked freely within these given quadrant constraints\footnote{These constraints are merely introduced to enhance the graphical display, avoiding overlap of the Hasse diagrams.}).

The Hasse diagram of the dual order $\alpha^{*}$ can be obtained (up to isomorphism) from the Hasse diagram of the order $\alpha$ (drawn in the above way) via the following process of \emph{Reflection} and \emph{Complementation}. \emph{Reflection} transforms the Hasse diagram of the original order into a proper Hasse diagram of the dual order (where mere complementation on the original Hasse diagram would not achieve this goal due to the complementation formation's ``left-to-right" nature, conflicting with a Hasse diagram's vertical setup). \emph{Complementation} represents the interchanging of the series and parallel operations. The dual SP-order construction proceeds as follows: 

\begin{itemize}
\item{({\bf Reflection}) Reflect each point $x_i$ with respect the bisector $y = x$}
\item{({\bf Complementation}) Add directed edges between reflected pairs $x_{i}^{*},x_{j}^{*}$ in case $(x_i,x_j)$ is an unrelated pair in the Hasse diagram of $\pol_{\Psi}$ (in other words: create the complement graph on reflected pairs). An edge between reflected elements \emph{points from} $x_{i}^{*}$ to $x_{j}^{*}$ in case $x_{i}^{*}$ is \emph{below} the $x_{j}^{*}$ in the dual Hasse diagram, or, equivalently, in case, for an unrelated pair $x_i$ and $x_j$, the element $x_{i}$ occurs to the left of $x_{j}$ in the original Hasse diagram.}
\end{itemize}

\begin{example} Hasse diagram for $\alpha = ((x_1 \otimes x_2) \para x_3) \otimes x_4)$ and Hasse diagram of the dual $\alpha^{*} = ((x_1 \parallel x_2)\otimes x_3) \parallel x_4)$.\end{example}

\begin{figure}[!ht]
    \centering
    \caption{Separating the state spaces}
        \subfloat[\emph{Reflection}: forming the Hasse diagram of $\alpha^{*} = ((x_1 \parallel x_2)\otimes x_3) \parallel x_4)$]{\includegraphics[width=0.4\columnwidth]{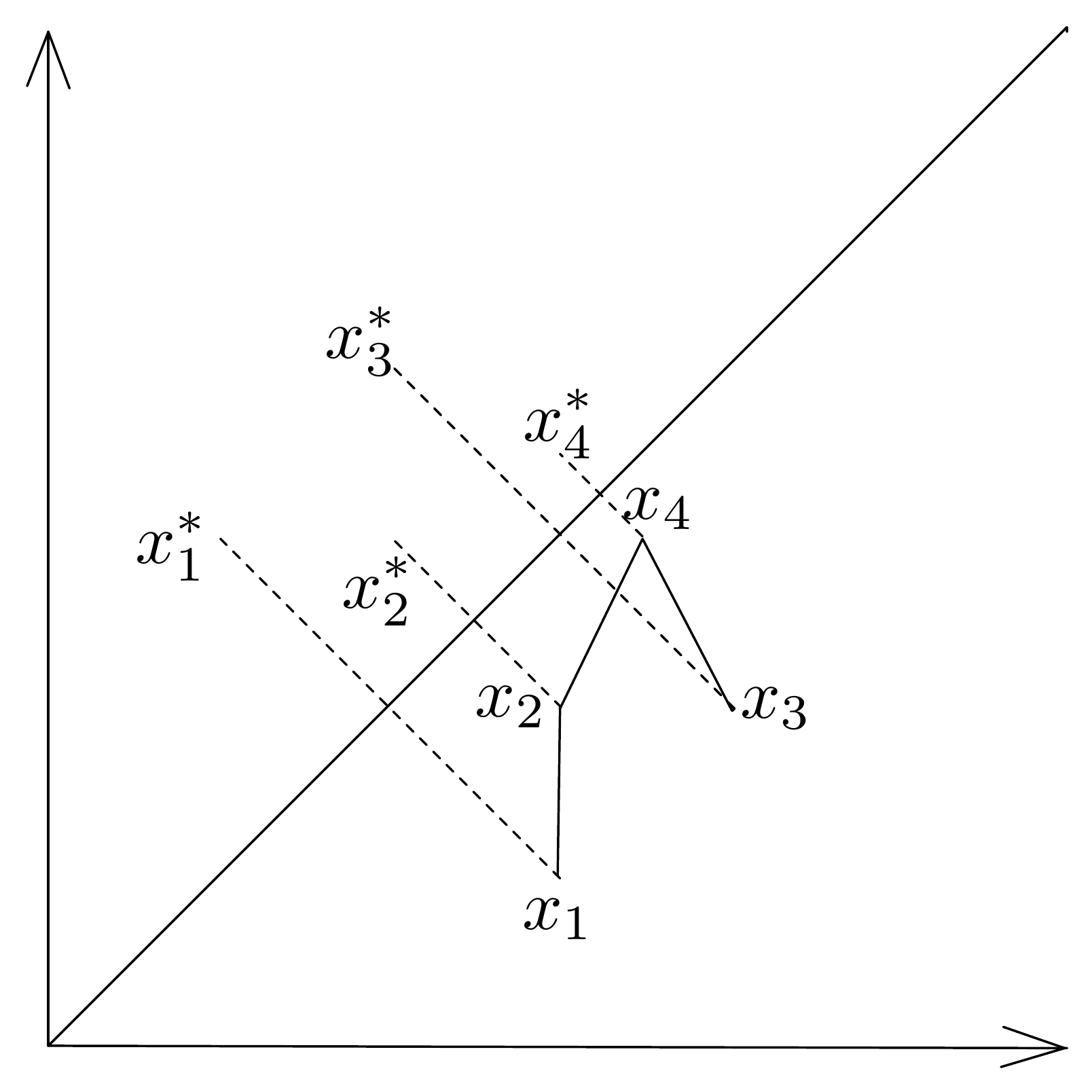}}
        \qquad \qquad \qquad
        \subfloat[\emph{Complementation}: forming the Hasse diagram of $\alpha^{*} = ((x_1 \parallel x_2)\otimes x_3) \parallel x_4)$]{\includegraphics[width=0.4\columnwidth]{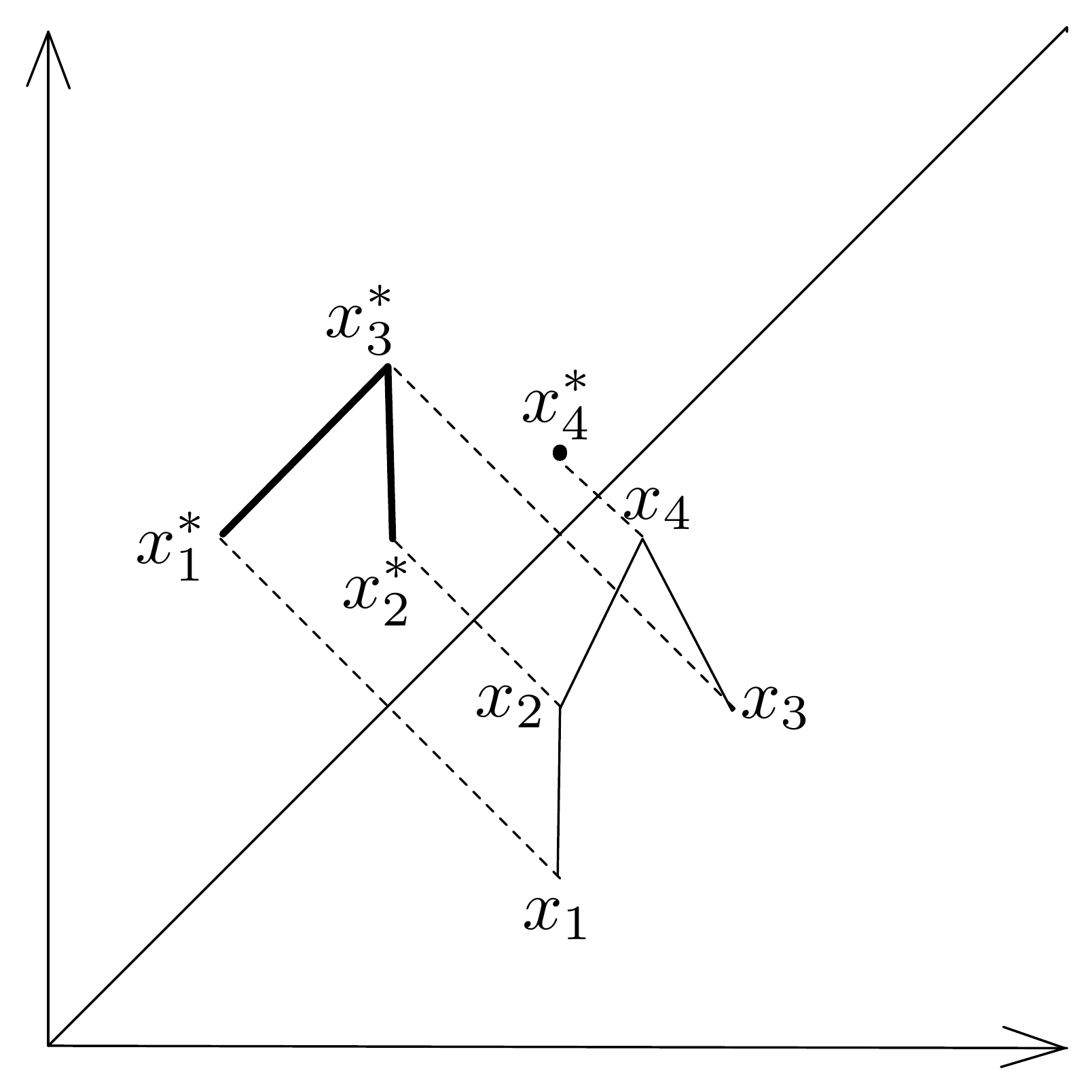}}
    \label{fig:subfigname}
\end{figure}

\subsection{Entropy conservation equality for double state spaces} \label{entropydouble}
As was the case in Section \ref{entropy}, we equate labels with ``quantitative information" and indices with ``positional information" and we will show the following result for double state spaces.

\begin{definition}
Given a SP-order $\alpha$ determined by SP-expression $\Psi$. Let $\alpha^{*}$ be the dual order determined by $\Psi^{*}$. We define the \textbf{quantitative entropy} (label-entropy) $H_q$ and the \textbf{positional entropy} (index-entropy) $H_{p}$ as follows: $H_q = log_2(|R(\alpha)|)$ and $H_p = log_2(|R(\alpha^{*})|)$.

The pair $(H_q,H_p)$ is the \textbf{entropy of a double state space} $(R(\alpha),R(\alpha^{*}))$. The \textbf{maximum entropy} is the entropy-pair of a state space $R(\Delta_n)$ over the discrete order of size $n = |\alpha|$, given by 

$$H_{max} = log_2(n!)$$

\end{definition}

The following theorem expresses entropy conservation between the components of the entropy-pair $(H_q,H_p)$ of a double state space $(R(\alpha),R(\alpha^{*}))$. 

\begin{theorem} (Entropy conservation) \label{duality} \\
Quantitative and positional entropy are inversely proportional: $H_{p} + H_{q} = H_{max}$
\end{theorem} 
\begin{proof} Given a SP-partial order $\alpha$ determined by SP-expression $\Psi$, where $|\alpha| = n$. The proof proceeds by induction on the number of operations $k$ in the formula $\Psi$. It suffices to show that $ |R(\alpha)| |R(\alpha^{*})| = n!$ \\

\noindent a) {\bf Base case $k = 0$:} we must have $\Psi = x_i$ for some variable $x_i$ where $i \in \{1,\ldots,n\}$ and $\alpha = \alpha^{*}$ (the discrete order on $x_i$). 
Hence $H_{p} \times H_{q} = 1 \times  1 = 1 = n!$ \\

\noindent b) {\bf Case $k > 0$ and $\Psi = \Psi_1 \otimes \Psi_2$:} then $X = X_1 \cup X_2$, $X_1 \cap X_2 = \emptyset$, $X_1 = Var(\Psi_1)$ and $X_2 = Var(\Psi_2)$. Let $l_1 = |X_1|$, $l_2 = |X_2|$ and $n = l_1 + l_2 = |X|$, where $l_1,l_2 \geq 1$. With some abuse of notation, we denote $|R(\alpha)|$ by $R(\Psi)$. Note that $H_{p} \times H_{q} = |R(\alpha)| |R(\alpha^{*})| = |R(\Psi)| |R(\Psi^{*})|$, hence:
\begin{eqnarray*}
|R(\Psi)| |R(\Psi^{*})| &=&  |R(\Psi_1 \otimes \Psi_2)| |R(\Psi_1^{*} \parallel \Psi_2^{*})| \\
&=& |R(\Psi_1)||R(\Psi_2)| {n \choose l_1} |R(\Psi_1^{*})||R(\Psi_2^{*})| \\
&=&  {n \choose l_1}|R(\Psi_1)|  |R(\Psi_1^{*})| |R(\Psi_2)||R(\Psi_2^{*})| \\
&=& \frac{n!}{l_1! (n-l_1)!} l_1!l_2! \mbox{\hspace*{1 cm} (Induction hypothesis)} \\
&=& n!
\end{eqnarray*}

\noindent c) {\bf Case $k > 0$ and $\Psi = \Psi_1 \parallel \Psi_2$:} the proof of case c) proceeds similar to that of case b). 

\end{proof}

\begin{remark} Entropy is in a sense a measure of ``disorganization'' (degrees of freedom), hence in the context of global state preservation Theorem \ref{duality} can be interpreted as: \textbf{``quantitative order gained is proportional to positional order lost\footnote{Which may bear some relation to the messy office argument advocating a chaotic office where nothing is in the right place yet each item's place is known to the owner, over the case where each item is stored in the right order and yet the owner can no longer locate the items.}''}  We note that global state conservation for MOQA-computations \cite{sch1} means that the global states must have the same cardinality (as multisets) throughout the computation, hence $H_{p} + H_{q} = H_{max}$ remains constant throughout the computation when one global state is transformed into a second. 
\end{remark}

\section{Conclusion and future work}
We have established a ``denotational" version of entropy conservation for comparison-based sorting for which the sum of the entropies (positional and quantitative) remains constant when proceeding from the input to the output collection. We generalized entropy conservation to global states over SP-orders, for comparison-based computations transforming a global state into global state of the same cardinality (such as the operations of \cite{sch1}). The question remains whether there exists an ``operational" mechanism that not only transforms the underlying order $\alpha$ into $\alpha^{*}$ but \emph{also} transform the labeling $l$ into a ``dual labeling'' $l^{*}$ such that $(\alpha^{*},l^{*})$ forms once again a topological sort. This has been achieved for the case of sorting. Figure \ref{Fig:sorting-history-general-paired} illustrates that permutations $\sigma$, i.e. labelings of the discrete order are transformed into labelings $\sigma^{-1}$ of the dual order, i.e. the linear order. A general operational mechanism for comparison-based computation over SP-orders (using a suitable fragment of the MOQA language of \cite{sch1}) will be explored in future work, in which we establish an ``entropic duality theorem" coupling the computations with a dual computation \emph{per labeling}. The latter computation effects an increase in positional entropy proportional to the decrease in quantitative entropy effected by the original computation. We refer to algorithms satisfying this type of entropic coupling as ``diyatropic". Comparison-based algorithms form one of the most thoroughly studied classes in computer science \cite{knu}. Our results indicate that more aspects still remain to be explored in this area.

\end{document}